\documentclass[11pt,a4paper]{article}
\usepackage{workingpaper}

\newcommand{\PaperTitle}{Multidimensional Online Resource Allocation for Large Language Model Services under Nonstationary Demand}
\title{\vspace{-1.5cm}{\fontfamily{lmr}\fontsize{16}{19}\selectfont\bfseries\PaperTitle\par}}
\author{%
    \begin{tabular}{@{}c@{\hspace{3em}}c@{}}
        Hao Wang & Zihao Li \\[-1pt]
        {\footnotesize haowangsom@ustc.edu.cn} & {\footnotesize zihao\_li@ustc.edu.cn}
    \end{tabular}\\[3pt]
    {\footnotesize Faculty of Business for Science and Technology, School of Management,}\\
    {\footnotesize University of Science and Technology of China}\\[10pt]
    Zhenzhen Yan \\[-1pt]
    {\footnotesize yanzz@ntu.edu.sg}\\[3pt]
    {\footnotesize School of Physical and Mathematical Sciences \& Nanyang Business School,}\\
    {\footnotesize Nanyang Technological University}
}
\date{}
\hypersetup{pdftitle={\PaperTitle},pdfauthor={Hao Wang, Zihao Li, Zhenzhen Yan}}

\begin{document}
\maketitle
\thispagestyle{empty}
\begin{abstract}
\noindent We study multidimensional online generalized assignment under nonstationary arrivals, motivated by resource allocation for large language model agent services.
Providers may commit model and tool allowances to agent tasks upon admission, reserving budgets for their internal execution loops.
Each request is rejected or accepted with an irrevocable assignment to one compatible server, earning reward and consuming a vector of its nonreplenishing budgets.
Arrivals are independent with known, potentially time-varying type distributions.
Scalar reductions can discard feasible combinations of requests with complementary resource requirements.
To maximize expected reward, we develop Vector-Aware Best-Fit, which combines linear-programming-guided candidate routing with calibrated acceptance while retaining vector feasibility.
Relative to the expected clairvoyant offline optimum, its guarantee improves on scalar reduction in the multidimensional case and has the best possible inverse-dimensional order.
With an appropriate acceptance parameter, the proposed policy attains the optimal competitive ratio on all-heavy instances, where every feasible assignment consumes more than half of at least one resource budget.
We further develop two heuristics that retain its performance guarantee.
Computational results show that the heuristics perform well against benchmarks in most tested settings, especially when capacity is scarce relative to demand.

\vspace{1em}
\vspace{0.5em}
\end{abstract}

\onehalfspacing
\hrule
\vspace{4em}

\section{Introduction}
\label{sec:introduction}

Large language model (LLM) services increasingly support multi-step agent tasks, such as software development, research, and enterprise automation. Completing these tasks may involve repeated model invocations, extensive context processing, output generation, and external tool use. Consequently, admitting an agent session can require commitments against several distinct resource budgets, including input-token, output-token, tool-call, and model-call allowances. A provider must decide whether to admit each arriving task and which compatible service pool should serve it, while preserving limited resources for future requests whose identities and values remain uncertain.

We model LLM-agent admission as a multidimensional online generalized assignment problem (MOGAP). Each arriving agent session is a request, and each service pool is a server to which that request may be assigned. A pool is compatible with a session if it supports the required model, tools, and execution environment. Each pool has a vector of available allowances over a finite planning horizon, with separate coordinates for resources such as input tokens, output tokens, tool calls, and model calls. Assigning a session to a compatible pool earns a reward, representing its service value or revenue, and commits a vector of allowances against that pool’s budgets. Both the reward and the allowance vector may depend on the selected pool, reflecting differences in service quality, model configuration, or execution requirements. For example, a coding session may reserve input tokens for repository context, output tokens for generated code, tool calls for testing, and model calls for iterative debugging.

The provider must immediately reject each arriving session or assign it irrevocably to one compatible pool. The assignment is feasible only if the pool has sufficient remaining allowance in every resource coordinate; admitting the session then deducts its entire allowance vector from those budgets. Thus, the ``generalized assignment" structure captures the choice among heterogeneous pools, while the ``multidimensional" structure captures the simultaneous resource commitments within the selected pool. The problem is online because decisions precede the realization of future requests: session types arrive independently according to known distributions that may vary over time. The objective is to maximize expected total reward while ensuring that cumulative commitments never exceed any pool's resource budgets. We treat the allowance vector as specified and known at admission, rather than as the session's uncertain realized usage. These allowances are nonreplenishing within the planning horizon, so the model addresses admission and budget commitment; physical-resource scheduling, budget replenishment, and reclamation after execution are outside its scope.

This problem presents a fundamental allocation challenge: admitting a request changes several resource balances simultaneously, and the value of the remaining capacity depends on the resource profiles of future requests. A decision that appears attractive based on current reward or aggregate utilization may exhaust a critical resource and prevent subsequent high-value admissions. Conversely, rejecting a request because of a conservative aggregate capacity measure may leave resources unused even when complementary tasks could be served together. Effective admission control must therefore account jointly for multidimensional feasibility, heterogeneous request sizes, and uncertainty in the evolving demand mix.

The first challenge is to exploit resource complementarity without violating any individual budget. Agent tasks can have markedly different resource profiles. A repository-analysis task may require substantial input context, whereas a drafting task may consume more output tokens; a tool-intensive workflow may draw heavily on tool-call allowances. These resources are not interchangeable, and their requirements cannot generally be represented faithfully by a single task size. For example, two requests with normalized consumption vectors $(0.6,0.2)$ and $(0.2,0.6)$ can both be admitted under a budget vector $(1,1)$. A maximum-coordinate scalarization assigns each request a size of $0.6$, incorrectly excluding their joint admission under a unit scalar budget. Although this reduction preserves feasibility, it sacrifices combinations that use different resources efficiently. The challenge is to retain these combinations while controlling the risk that an admission exhausts any one coordinate.

The second challenge is to accommodate large, indivisible commitments. An individual session may reserve a substantial fraction of a pool's allowance in one or more dimensions. Such requests fall outside the small-consumption regime underlying many fluid and large-budget approximations. Their admission can sharply alter the opportunities available to later requests, making immediate reward an insufficient guide to allocation. Moreover, being large in one coordinate does not imply that two requests are mutually incompatible: requests that are heavy in different coordinates may still fit together. A policy must therefore manage large commitments through their actual vector requirements rather than treating every heavy request as occupying the entire pool.

The third challenge is to allocate resources under nonstationary demand. The relative frequencies of different request types may change across the planning horizon, so early observations need not represent later demand. A provider may need to preserve output-token allowances for a later concentration of drafting requests or tool-call allowances for an anticipated increase in coding workflows. We model this setting through independent arrivals with known, potentially time-varying type distributions. Future realizations remain unknown, even though their distributions are available for planning. The resulting decision problem requires balancing current reward against future allocation opportunities under a changing demand composition.

Existing online allocation results address important parts of this problem, but do not directly provide guarantees for this combination of vector budgets, nonstationary arrivals, and arbitrary consumption fractions. 
\cite{naoriRaz2019} study multidimensional online assignment and allow both large and small requests. Their limitation for our setting is the arrival model, rather than the resource representation: a fixed collection of requests arrives in a uniformly random permutation. Their sampling-based analysis exploits that random order. When request-type distributions change with time, an early segment need not represent later demand, and their competitive guarantee does not directly transfer. The random-order and independent nonstationary models have different information structures; neither contains the other.

\cite{jiangMaZhang2024} \wh{(JMZ)} address single-dimensional online generalized assignment under independent, potentially nonstationary arrivals and arbitrary consumption fractions. Their framework therefore accommodates changing demand and large requests, but imposes only one capacity constraint per server. Applying their policy to a maximum-coordinate scalar surrogate yields a feasible multidimensional policy, at the cost of excluding complementary resource combinations. We quantify this approach through a competitive guarantee of $(1/[D(3+e^{-2})])$, where $D$ is the maximum retained resource dimension at a server \wh{after removing dominated and duplicate resource coordinates}. Thus, a direct scalar reduction provides a useful benchmark, but leaves room for policies that exploit vector feasibility.

These observations motivate our study of the multidimensional online generalized assignment problem (MOGAP) under independent nonstationary arrivals and arbitrary consumption fractions. The central question is \emph{how to convert a fractional allocation plan into irrevocable online assignments while retaining resource complementarity and guaranteeing feasibility}. This is more demanding than applying a single-resource acceptance rule separately to each coordinate: a request must be admitted as a complete assignment, and its acceptance depends on the joint resource state at its selected pool. Likewise, expected capacity constraints alone cannot prevent budget violations on individual arrival paths. A successful policy must coordinate routing and acceptance across all resource dimensions while preserving sufficient opportunities for future demand.

\subsection{Our Contributions}
\label{sec:contributions}
We address this challenge through a vector-aware online policy, competitive analysis, and computational evaluation. Our contributions are fourfold.

\paragraph{Formulation of MOGAP under nonstationary demand.}
We formulate MOGAP under independent arrivals with known, time-varying type distributions and arbitrary resource-consumption fractions. The model jointly captures pool compatibility, assignment-dependent rewards, and simultaneous commitments against multiple budgets.
We quantify the cost of maximum-coordinate scalarization: an $\alpha$-competitive single-dimensional assignment policy yields an $\alpha/D$-competitive policy for the  multidimensional problem, where $D$ is the maximum retained resource dimension at any pool (Lemma~\ref{lem:oracle-baseline}).
Applying \cite{jiangMaZhang2024} gives the benchmark 
$1/[D(3+e^{-2})]$ (Corollary~\ref{cor:scalar-jmz}).
This benchmark provides a precise reference for evaluating the benefit of retaining vector feasibility.

\paragraph{A vector-aware policy with improved guarantees.}

We develop \wh{Vector-Aware Best-Fit (VABF)} from an LP relaxation that retains coordinatewise budget constraints and adds valid inequalities for assignments consuming more than half of a coordinate's capacity. \wh{VABF} combines LP-guided candidate routing with calibrated acceptance based on the distribution of actual vector loads. With an appropriate acceptance target \wh{(an input parameter specifying a proposal's intended acceptance probability)}, its competitive ratio is
\[
R(D)=\max\left\{\frac{1}{2D+1+e^{-2D}}, \frac{4}{7D+8.056}\right\}. 
\]
The guarantee holds relative to the strengthened LP and hence to the expected clairvoyant offline optimum (Theorem~\ref{thm:discrete-main}).
It equals the scalar benchmark when $D=1$ and strictly improves upon it for every $D\ge2$. \wh{VABF} retains feasible combinations of complementary requests that the scalar surrogate would exclude.


\paragraph{Performance limits and an optimal result for all-heavy instances.}
We establish that no randomized nonanticipatory policy can guarantee more than $1/(D+1)$ of the expected offline optimum, even when every feasible assignment consumes more than half of at least one resource budget (Theorem~\ref{thm:limit-heavy}). On the all-heavy subclass, where every feasible assignment consumes more than half of at least one resource budget at the selected server, \wh{VABF} achieves the matching ratio with an appropriate acceptance target. Thus, the general guarantee has the best possible inverse-dimensional order, and the all-heavy result is optimal. We separately establish an upper bound of $1/(D+\wh{2})$ relative to the strengthened LP (Theorem~\ref{thm:lp-online-limit}), distinguishing the limitation against the offline allocation benchmark from the limitation against the relaxation.


\paragraph{Heuristics that retain guarantees and computational evaluation.}
We develop two heuristics to improve utilization while preserving the theoretical baseline. \wh{VABF with tuning (VABF-T)} increases routing probabilities on the LP solution's support and calibrates feasible acceptance parameters to protect baseline expected reward. \wh{VABF with tuning and recovery (VABF-TR)} adds a recovery step that admits additional requests while retaining \wh{VABF-T}'s guarantee. Experiments on instances calibrated from $3,798$ multi-call TraceLab sessions compare the policies under input-token, output-token, tool-call, and model-call budgets. The heuristics perform well against the tested benchmarks in most settings, especially when capacity is scarce relative to demand. These experiments evaluate the proposed allowance-allocation model; the budgets, rewards, and independently sampled arrival laws are specified components of the experimental construction.


Together, these results show how multidimensional admission control can exploit resource complementarity while providing performance guarantees under nonstationary demand. They also identify the unavoidable dependence on resource dimension and translate the theoretical policy into practical allocation heuristics.

\hyperref[tab:result-scopes]{Table~\ref*{tab:result-scopes}} summarizes our lower and upper bounds, with the benchmark and supporting theorem specified for each result.
\wh{Here $\OPT$ denotes the expected clairvoyant optimum and $Z_{\mathrm{LP}}$ the strengthened LP optimum.}
\hyperref[fig:vector-scopes]{Figure~\ref*{fig:vector-scopes}} illustrates these bounds for {$D=1,\ldots,10$} alongside the scalar JMZ reduction.
For reference, the figure also includes the guarantee of \citet{naoriRaz2019} under random order, where a fixed input arrives in a uniformly random permutation.
The Naori--Raz curve uses a different arrival model and is not a guarantee under our independent nonstationary arrivals (Section~\ref{sec:model}).

\begin{table}[ht]
\centering \caption{{Competitive ratio lower and upper bounds.}}
\label{tab:result-scopes}
\begin{tabularx}{\textwidth}{@{}llXX@{}}
\toprule
\leavevmode{Instance class} & {Benchmark} & \leavevmode{\wh{VABF} lower bound} & \leavevmode{Uniform upper bound}\\
\midrule
General & $\OPT$ & $R(D)$ (Thm.~\ref{thm:discrete-main}) & $1/(D+1)$ (Thm.~\ref{thm:limit-heavy})\\
General & $Z_{\mathrm{LP}}$ & $R(D)$ (Thm.~\ref{thm:discrete-main}) & $1/(D+\wh{2})$ (Thm.~\ref{thm:lp-online-limit})\\
\leavevmode{All-heavy} & $\OPT$ & $1/(D+1)$ (Thm.~\ref{thm:limit-heavy}) & $1/(D+1)$ (Thm.~\ref{thm:limit-heavy})\\
\bottomrule
\end{tabularx}
\end{table}

\begin{figure}[htbp]
\centering
\includegraphics[width=.90\textwidth]{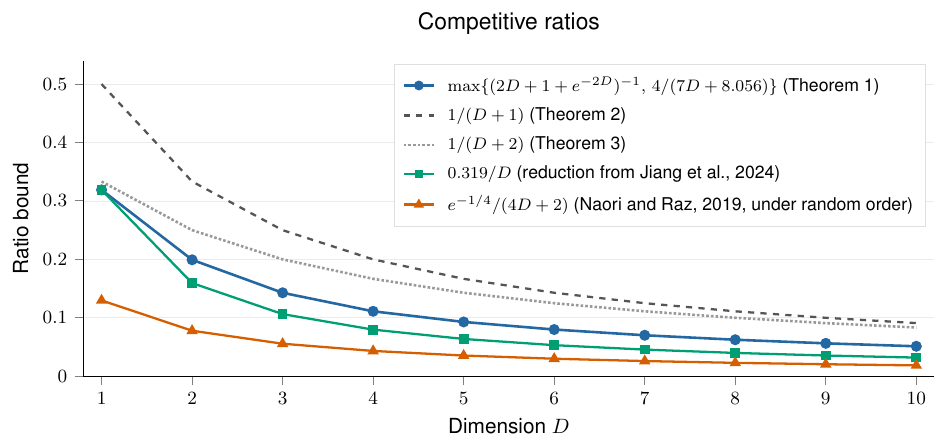}
\caption{{Competitive ratios and upper bounds.
Solid curves show \wh{VABF} (Theorem~\ref{thm:discrete-main}), the scalar reduction of \citet{jiangMaZhang2024}, and {$e^{-1/4}/(4D+2)$} from \citet{naoriRaz2019} under a uniformly random permutation of a fixed input, included only for reference.
The dashed {$1/(D+1)$} curve is our {$\OPT$}-relative upper bound (Theorem~\ref{thm:limit-heavy}).
The gray dotted {$1/(D+\wh{2})$} curve is our {$Z_{\mathrm{LP}}$}-relative upper bound (Theorem~\ref{thm:lp-online-limit}).
Our bounds and the scalar JMZ reduction use the nonstationary arrival model in Section~\ref{sec:model}.}}
\label{fig:vector-scopes}
\end{figure}

\subsection{{Related Work}}
\label{sec:related-work}

{\paragraph{Single-dimensional OGAP and online knapsack.}
The single-dimensional online generalized assignment problem (OGAP) has one capacity constraint per server; online knapsack is the special case with one server.
\zihao{Under random arrival order, \citet{albersEtAl2021}
obtain asymptotic competitive ratios of $1/6.99$
for OGAP and $1/6.65$ for online knapsack.
\citet{klimmKnaack2026} subsequently improve both
guarantees to $(1-\ln 2)/2 \approx 1/6.52$
in the same arrival model.}
In the single-dimensional counterpart of our nonstationary arrival model, \citet{jiangMaZhang2024} use an LP-based reduction from assignment to single-dimensional knapsack decisions at individual servers.
Their guarantee {$(3+e^{-2})^{-1}$} permits arbitrary consumption fractions revealed before assignment and is tight relative to their standard ex ante LP relaxation.
\citet{alaeiEtAl2013} consider a different information structure: random consumption is revealed after assignment.
Their expected-instance LP and generalized magician yield {$1-1/\sqrt{k}$} when every possible consumption is at most {$1/k$} of capacity, for an integer {$k\ge2$}.
\citet{wangYanBei2022} analyze LP-based randomized allocation and periodic re-solving for integral demands and capacities under known independent and identically distributed (i.i.d.) arrivals.
These results distinguish the effects of arrival information and consumption sizes; our model accommodates nonstationary demand and simultaneous consumption in multiple resource dimensions at the selected server.}

{\paragraph{Multidimensional OGAP.}
\citet{naoriRaz2019} study multidimensional assignment under a uniformly random permutation of a fixed input and guarantee {$e^{-1/4}/(4D+2)$} of the offline optimal reward.
They also establish that inverse dependence on dimension is unavoidable in that model.
Their impossibility construction supplies the compatible-group geometry used in our dimensional upper bound.
Our nonstationary arrival model permits request-type distributions to vary by decision epoch (Section~\ref{sec:model}).
Neither arrival model contains the other, and the algorithms and hardness results require analyses under their respective assumptions.
\zihao{\citet{dickersonEtAl2019} study online assignment with
multidimensional resource budgets and time-varying
arrival distributions.
For assignments involving multiple non-integral resources,
their competitive guarantees require lower bounds on
the corresponding resource budgets.
Without such budget requirements, their guarantee applies
to a restricted setting in which each assignment consumes
only one resource, arrival distributions are identical
across epochs, and servers have no deadlines.
In contrast, our model allows arbitrary multidimensional
consumption fractions and nonstationary arrivals without
requiring large resource budgets.}}

{\paragraph{Online packing LPs.}
MOGAP is also a structured online packing problem: an assignment consumes resources within one server, and assignments to that server compete for its resource coordinates.
Much of the online packing LP literature seeks guarantees close to the offline optimum when capacities are large relative to individual consumptions.
Under random arrival order, training and dual-pricing methods \citep{feldmanEtAl2010,agrawalEtAl2014,molinaroRavi2014} and primal LP rounding \citep{kesselheimEtAl2018} quantify the capacity requirements and convergence toward a competitive ratio of one.
\citet{feldmanEtAl2010} additionally require small individual contributions to the offline optimal value; \citet{kesselheimEtAl2018} also provide guarantees for small capacity ratios.
Another objective is additive regret, the expected shortfall from the offline optimal reward.
For a fixed number of resource constraints, \citet{liYe2022} obtain {$O(\log T\log\log T)$} regret under i.i.d.\ input from an unknown distribution, capacities proportional to the horizon, and their regularity and nondegeneracy assumptions.
Our analysis permits arbitrary consumption fractions and identifies the resulting inverse-dimension scale under nonstationary arrivals.}

\paragraph{Online contention resolution schemes.}
Online contention resolution schemes (OCRSs), introduced by \citet{feldmanSvenssonZenklusen2016}, convert fractional solutions into feasible online selections of independently activated elements.
Selection guarantees conditional on an element's activation translate fractional value into expected online reward.
Related knapsack schemes also depend on arrival and information assumptions.
\citet{maMacRuryStein2025} obtain {$1/3$} selectability when the order is chosen uniformly between a fixed order and its reverse.
\citet{yoshinagaKawase2025} obtain {$1/3$} acceptance when sizes are revealed after insertion and overflowing items are rejected; their hard-capacity reward reduction requires values independent of the random type.
Our setting uses the given order of decision epochs and permits rewards to depend on the realized request type.
\wh{VABF} follows the OCRS framework applied to single-dimensional OGAP by \citet{jiangMaZhang2024}.
We extend the one-dimensional utilization distribution of JMZ's OCRS to the joint distribution of actual vector load, while using a scalar threshold on maximum coordinate load to control acceptance.
Our new analysis combines cumulative scalar load, accepted heavy-assignment count, and vector complementarity to certify the guarantees for multidimensional capacities.

\paragraph{Operations Research for LLMs.}
Recent operations research studies examine scheduling and routing for LLM inference under computation and memory constraints.
\citet{aoEtAl2025Inference} develop fluid-guided scheduling under memory constraints.
\citet{nieSiZhou2026Inference} analyze queueing stability with computation and KV-cache constraints.
\citet{chenYeZhou2026Routing} study routing across decode workers under batch-size and KV-cache limits.
Our decision concerns horizon-level budget commitments: for each arriving task, the provider either rejects it or accepts it with an assignment to one compatible service pool.
This differs from scheduling individual calls and releasing physical resources during execution.

Sections~\ref{sec:model}--\ref{sec:limits} present the model and single-dimensional reduction, \wh{VABF} and its competitive guarantee, and competitive-ratio limits, respectively.
Appendices~\ref{app:core-certificates}--\ref{app:lp-hardness} provide the supporting proofs and implementation details.
Section~\ref{sec:computations} evaluates the two heuristics, benchmark comparisons, and parameter sensitivity.
Appendix~\ref{app:computational-details} gives the experimental construction, policy details, and complete results.

\subsection{Notations}
\label{sec:notations}
\wh{Vectors are typeset in boldface}, and inequalities between vectors are interpreted coordinatewise.
For a vector $\mathbf{x}=(x^d)_{d\in\mathcal D}$, its infinity norm is $\|\mathbf{x}\|_\infty=\max_{d\in\mathcal D}|x^d|$; for the nonnegative consumption and load vectors in our model, this equals the maximum coordinate.
The symbols $\mathbf0$ and $\mathbf1$ denote the all-zero and all-one vectors of the appropriate dimension.
We use $\Prob$ and $\E$ for probability and expectation, respectively, and $\one_{\{A\}}$ for the indicator of an event $A$.
For a finite set $\mathcal A$, $|\mathcal A|$ denotes its cardinality, and $\lceil x\rceil$ denotes the smallest integer no less than $x$.
\wh{The positive part of $x$ is $x^+=\max\{x,0\}$.}


\section{{Model and Single-Dimensional Reduction}}
\label{sec:model}
\subsection{{Model}}
\label{sec:model-formulation}
In the motivating LLM-agent application, each request represents an agent session, and each server represents a service pool with a vector of allowance budgets for the planning horizon. Each compatible session--pool pair has a known reward and allowance vector, specifying commitments such as input tokens, output tokens, tool calls, and model calls. Upon arrival, the provider either rejects the session or irrevocably assigns it to one compatible pool with sufficient remaining allowance in every resource coordinate. Acceptance earns the corresponding reward and commits the specified allowances. These budgets do not replenish within the horizon; runtime scheduling and the release of physical resources after execution are outside the model.

There is a finite set of servers $U$, a finite set of request types $V$, and a compatibility graph $E\subseteq U\times V$.
There are $T\ge1$ decision epochs, indexed by $t=1,\ldots,T$, with at most one active request per epoch.
At epoch $t$, a random type $V_t$ is revealed with probabilities
\begin{equation}
 p_{vt}=\Prob(V_t=v),\qquad p_{vt}\ge0,\qquad
 \sum_{v\in V}p_{vt}=1 \quad(t=1,\ldots,T).
 \label{eq:arrival-probabilities}
\end{equation}
{The distributions may vary arbitrarily across decision epochs, allowing nonstationary arrivals.}
The variables $V_1,\ldots,V_T$ are independent.
The horizon $T$, the probability array, and all assignment attributes are known in advance.
Future type realizations are unknown.
The type set may include a null type with no compatible server to represent an inactive epoch; such an outcome is always rejected.

Assigning a type-$v$ request to server $u$ earns reward $w_{uv}\ge0$ and consumes a deterministic vector $\mathbf{s}_{uv}=(s_{uv}^d)_{d\in\mathcal D_u}$, where $\mathcal D_u$ is the finite set of resource coordinates at server $u$.
Thus $u$ indexes assignment destinations, whereas $d\in\mathcal D_u$ indexes resource coordinates within the selected destination.
Positive capacities are normalized to one, and infeasible edges are removed.
Each coordinate size lies in $[0,1]$, and every feasible assignment has positive consumption in at least one coordinate: $\max_d s_{uv}^d>0$ for every $(u,v)\in E$.
Individual coordinates may be zero.
Coordinates unused by all retained edges are deleted.

{An edge {$(u,v)\in E$} is \emph{heavy} if {$\max_{d\in\mathcal D_u}s_{uv}^d>1/2$}, and is \emph{heavy in coordinate {$d$}} if {$s_{uv}^d>1/2$}.
An edge is \emph{light} if {$\max_{d\in\mathcal D_u}s_{uv}^d\le1/2$}; in particular, a maximum size of exactly one half is light.
A heavy assignment or proposal uses a heavy edge.
This classification is specific to the type--server pair: a type can have a heavy edge to one server and a light edge to another.
References to heavy requests are relative to the specified server.}

{An instance is \emph{all-heavy} if every feasible edge in {$E$} is heavy.
Different edges may be heavy in different coordinates; an all-heavy instance does not require every coordinate of every edge to exceed one half.}

A request must be assigned to at most one compatible server immediately upon arrival or rejected.
Allocations are irrevocable, and resources do not replenish.
Feasibility requires that the total consumption of every server-coordinate pair be at most one on every realization.

We preprocess each server's coordinates by dominance.
Write $E_u=\{(u,v)\in E\}$.
The profile of coordinate $d$ is $(s_{uv}^{d})_{(u,v)\in E_u}$, using all retained assignment edges.
First retain one representative of each set of identical profiles.
Among the distinct profiles, delete a coordinate $d$ if another coordinate $k$ satisfies $s_{uv}^{d}\le s_{uv}^{k}$ for every $(u,v)\in E_u$.
This preprocessing preserves every feasible allocation and each edge's maximum scalar size; Appendix~\ref{sec:dimension-preprocessing} gives the argument and computational cost.
Hereafter $\mathcal D_u$ denotes the coordinate set after preprocessing.
Define
\[
D=\max\bigl(\{1\}\cup\{|\mathcal D_u|:u\in U\}\bigr).
\]
The dimension $D$ concerns the budget vector within a server, whereas $|U|$ is the number of servers.
In particular, a single-dimensional instance can still have multiple servers.

Write
\begin{equation}
 \mu_v=\sum_{t=1}^T p_{vt}.
 \label{eq:arrival-masses}
\end{equation}
Thus $\mu_v$ is the expected number of type-$v$ requests over the whole horizon.
Types with $\mu_v=0$ can be removed.

\emph{Competitive ratio.}
Let {$\OPT(I)$} denote the expected clairvoyant offline optimum (or offline optimum, if no ambiguity) for instance {$I$}: for each realized arrival sequence, the offline decision maker maximizes reward under the same assignment and capacity constraints, and {$\OPT(I)$} averages this optimal value over the arrivals.
\wh{Write $\ALG(I)$ for the policy's random total reward on instance $I$.}
A competitive ratio {$r$} means {$\E[\ALG(I)]\ge r\OPT(I)$} for every instance in the stated class; the expectation includes arrivals and policy randomization. 

\subsection{Single-Dimensional Reduction}
\label{sec:oracle-warmup}
\zz{A natural approach to MOGAP is to replace each resource-consumption vector by a scalar size and apply an existing single-dimensional allocation policy. We examine the maximum-coordinate reduction, which preserves feasibility but may discard valuable allocations that exploit resource complementarity.}

\zz{Recall that each server’s resource capacities are normalized to one. For every compatible edge $(u,v)$, define}
{\begin{equation}
 b_{uv}:=\|\mathbf{s}_{uv}\|_\infty.
 \label{eq:scalar-size}
\end{equation}}
For an instance {$I$}, its \emph{scalar surrogate} {$I_{\mathrm{sc}}$} retains the original servers, horizon, request-type distributions, compatibility graph, and rewards, but replaces each server’s vector of capacity constraints by a single unit-capacity constraint. Assigning a type-{$v$} request to server {$u$} consumes {$b_{uv}$} units of this scalar capacity. 
Let {$\OPT_{\mathrm{sc}}$} denote its expected clairvoyant offline optimum. 

\zz{This reduction is conservative. Since ($s_{uv}^{d}\leq b_{uv}$) for every coordinate, any scalar-feasible allocation is feasible under the original vector budgets. The converse need not hold: scalarization charges each assignment its largest consumption against a common budget, even when different assignments attain their largest consumption in different coordinates. For example, two assignments with vectors $(0.6,0.2)$ and $(0.2,0.6)$ jointly consume $(0.8,0.8)$ and are therefore vector-feasible. Their scalar sizes sum to $1.2$, so the scalar surrogate excludes their joint allocation.}

The following lemma bounds the loss from imposing this scalar capacity constraint.

\begin{samepage}
\begin{lemma}[Single-dimensional reduction]
\label{lem:oracle-baseline}
Let {$\mathcal A$} be \zz{a feasible nonanticipatory} algorithm with competitive ratio $\alpha$ on a class of single-dimensional instances of our model.
For any MOGAP instance {$I$} whose scalar surrogate belongs to this class, 
\zz{\begin{equation}
\frac{1}{D}\mathrm{OPT}(I)
\leq \mathrm{OPT}_{\mathrm{sc}}
\leq \mathrm{OPT}(I).
\end{equation}}
Implementing {$\mathcal A$}'s scalar-surrogate assignments in the original instance {$I$} gives \zz{a policy with expected reward at least $\alpha\OPT_{\mathrm{sc}}$.}
\end{lemma}
\end{samepage}
\emph{Proof idea.}
On each realized request sequence, partition an optimal vector allocation by a largest-consumption coordinate of each assignment.
Each group is scalar-feasible: at every server, its scalar load equals its consumption in the group's coordinate and is at most one.
At least one group earns {$1/D$} of the vector optimum.
Taking expectations and applying the guarantee of {$\mathcal A$} gives the result.
Appendix~\ref{sec:oracle-proof} provides the complete proof.

\zz{Lemma \ref{lem:oracle-baseline} provides a general performance guarantee for applying single-dimensional allocation algorithms to multidimensional instances through scalarization. The algorithm of \citet{jiangMaZhang2024} is a special case: substituting its competitive ratio into Lemma~\ref{lem:oracle-baseline} immediately yields the multidimensional guarantee, as stated in the following corollary. }

\begin{samepage}
\begin{corollary}[Reduction of \citet{jiangMaZhang2024}]
\label{cor:scalar-jmz}
{Applying the single-dimensional assignment policy of \citet{jiangMaZhang2024} (JMZ) to the scalar surrogate gives a feasible online policy for MOGAP with expected reward at least
{\begin{equation}
 \frac{1}{D(3+e^{-2})}\OPT(I).
 \label{eq:scalar-jmz-guarantee}
\end{equation}}
Its competitive ratio is therefore at least {$1/[D(3+e^{-2})]\approx0.319/D$}.}
\end{corollary}
\end{samepage}

\zz{Equation (\ref{eq:scalar-jmz-guarantee}) separates two sources of loss: the single-dimensional policy's guarantee relative to its own offline optimum, and the scalar surrogate’s potentially smaller allocation value. A vector-aware policy can avoid the second restriction by admitting complementary assignments whenever their original resource vectors fit. The next section develops such a policy and establishes a stronger guarantee under the original multidimensional budgets.}

\section{LP Relaxation and Vector-Aware Best-Fit Algorithm}
\label{sec:algorithm}
This section develops an online policy called Vector-Aware Best-Fit Algorithm (VABF) and establishes its competitive guarantee.
The algorithm first solves an LP relaxation that preserves the original resource-consumption vectors to capture multidimensional complementarity and, beyond this coordinatewise capacity extension, includes heavy-assignment count constraints to account for conflicts among edges that are heavy in the same coordinate.
\zihao{Given an optimal LP solution, VABF proceeds in two steps: routing and acceptance.}
\zihao{The routing step uses an optimal LP solution to select at most one candidate server for each arriving request. Once a candidate server is selected, the acceptance decision depends only on that server's resource state, allowing the acceptance rule to be analyzed separately for each server.}
To decide whether to accept a routed request, \wh{VABF} extends the one-dimensional consumption distribution of the knapsack OCRS \citep{jiangMaZhang2024} to the joint distribution of actual vector load at this server. 
\zihao{Given this joint distribution, VABF calibrates a scalar threshold on the maximum coordinate of the current load vector to achieve a prescribed acceptance probability while maintaining vector feasibility.}
\hh{The analysis identifies feasible choices of these acceptance targets and derives the resulting competitive-ratio guarantee.}

Section~\ref{sec:lp-relaxation} constructs the LP relaxation.
Section~\ref{sec:routing-admission} presents the complete algorithm based on its solution.
Section~\ref{sec:main-guarantee} proves the competitive-ratio lower bound by combining the coordinatewise heavy-assignment constraints with the different packing of vector assignments.

\Needspace{12\baselineskip}
\begin{samepage}{\interdisplaylinepenalty=10000
\subsection{LP Relaxation}
\label{sec:lp-relaxation}
{\wh{VABF} solves the following strengthened ex ante LP relaxation \wh{with decision vector $\mathbf z=(z_{uv})_{(u,v)\in E}$}:}
\begin{subequations}\label{eq:planning-lp}
\renewcommand{\theequation}{\theparentequation.\arabic{equation}}
{\begin{align}
Z_{\mathrm{LP}}=\max_{\wh{\mathbf z}\ge\wh{\mathbf0}}\quad&
 \sum_{(u,v)\in E}w_{uv}z_{uv} \tag{\theparentequation}\\
\text{subject to}\quad&
 \sum_{u:(u,v)\in E}z_{uv}\le\mu_v &&(v\in V),\label{eq:lp-supply}\\
&\sum_{v:(u,v)\in E}s_{uv}^dz_{uv}\le1
 &&(u\in U,\ d\in\mathcal D_u),\label{eq:lp-capacity}\\
&\sum_{v\in V_{ud}^{H}}z_{uv}\le1
 &&(u\in U,\ d\in\mathcal D_u).\label{eq:lp-heavy}
\end{align}}
\end{subequations}}\end{samepage}

Here $z_{uv}\ge0$ is the expected number of assignments on edge $(u,v)$, and the objective is their total expected reward.
Constraint \eqref{eq:lp-supply} bounds the total assignments of type $v$ by its expected arrival count $\mu_v$, and different servers compete for the same type.
Constraint \eqref{eq:lp-capacity} enforces the unit capacity of each server-coordinate pair.
Keeping the original vector coefficients allows the LP relaxation to account for requests that consume different coordinates and can fit together.

In \eqref{eq:lp-heavy}, {$V_{ud}^{H}=\{v:(u,v)\in E,\ s_{uv}^d>1/2\}$} is the set of types whose edges to server {$u$} are heavy in coordinate {$d$}.
Two such assignments cannot both be accepted on a feasible sample path, so their total expected assignment count is at most one.
We call these the heavy-assignment count constraints.


\hh{The following lemma establishes that the strengthened LP optimum $Z_{\mathrm{LP}}$ upper bounds the expected clairvoyant offline optimum.}

\begin{lemma}[LP upper bound]
\label{lem:planning-upper}
The LP relaxation satisfies $Z_{\mathrm{LP}}\ge\OPT$.
\end{lemma}{\emph{Proof idea.}
Every realized offline allocation respects realized supply and the capacity and heavy-assignment count constraints.
Taking expected assignment counts therefore gives an LP-feasible vector with the expected offline reward.
The full proof appears in Appendix~\ref{sec:lp-proof}.}

\subsection{{The Vector-Aware Best-Fit Algorithm}}
\label{sec:routing-admission}
\wh{VABF} first solves the LP relaxation and fixes an optimal solution $\wh{\mathbf z}$.
It then implements the assignment decision through candidate selection and acceptance.
Upon observing a type-$v$ request, it samples at most one candidate server $u$ with probability $z_{uv}/\mu_v$.
This sampling step, referred to as routing, generates a proposal, or a proposed edge $e=(u, v)$.
If a candidate is drawn, the policy applies the acceptance rule using that server's current resource load. 
\hh{The threshold-based acceptance rule uses the parameters $\{\gamma_{u,t}\}$, called \emph{acceptance targets}, indexed by server and epoch.}
Only an accepted proposal becomes an assignment and consumes the corresponding resource vector.

The unconditional probability \hh{of proposing} edge {$e=(u,v)$} \hh{at epoch {$t$}} is
{\begin{equation}
 q_{uv,t}=p_{vt}\frac{z_{uv}}{\mu_v},\qquad
 \sum_{t=1}^T q_{uv,t}=z_{uv}.
 \label{eq:routed-probability}
\end{equation}}
At a fixed server, each arrival gives one edge proposal or a null outcome.
These mutually exclusive outcomes are independent across {$t$}, because request types and routing coins are independent across requests.
We call this independent sequence of mutually exclusive outcomes a categorical stream.

Write {$\delta(u)=\{(u,v)\in E\}$} for the edges incident to {$u$}, and use {$z_e,w_e,s_e^d$} as abbreviations for their edge-indexed quantities.
Write {$b_e=b_{uv}$} for the scalar size defined in \eqref{eq:scalar-size}.
Since {$s_{uv}^d\le b_{uv}$} in every coordinate, cumulative accepted scalar size bounds every coordinate's consumption.
The heavy-edge condition in Section~\ref{sec:model} is equivalently {$b_e>1/2$}.

To define the selection rule after routing, fix a server and temporarily suppress {$u$}.
Its \emph{actual vector load}, or state, after epoch {$t$} is {$\mathbf L_t=\sum_{i\le t:\,\hh{e_i\text{ is assigned}}}\mathbf s_{e_i}$}, with {$\mathbf L_0=\mathbf0$}.
\hh{Here $e_i$ denotes the assignment edge at epoch $i$, and the sum includes only assignments to this server.}
Write \wh{$L_t^d$ for its load in coordinate $d$, and define}
{\begin{equation}
 M_t=\|\mathbf L_t\|_\infty,\qquad
 \nu_t(\boldsymbol{\ell})=\Prob(\mathbf L_t=\boldsymbol{\ell}).
 \label{eq:bf-vector-state}
\end{equation}}
The unconditional joint distribution {$\nu_t$} averages over arrivals and policy coins through epoch {$t$}; it is not conditioned on the realized history.
\wh{VABF} takes a sequence of acceptance targets {$\gamma_1,\ldots,\gamma_T$} as a parameter.
The target {$\gamma_t\in[0,1]$} specifies the probability of accepting edge {$e$} given that it is proposed at epoch {$t$}, for every {$e$} with {$q_{e,t}>0$}, averaging over earlier arrivals and policy coins.

{For a proposal on edge {$e$}, define its feasible states and their mass at each maximum load by
{\begin{equation}
 \mathcal S(e)=\{\boldsymbol{\ell}:\boldsymbol{\ell}+\mathbf s_e\le\mathbf1\},
 \qquad
 \pi_m^{e,t}=
 \sum_{\substack{\boldsymbol{\ell}\in\mathcal S(e)\\
                 \|\boldsymbol{\ell}\|_\infty=m}}\nu_{t-1}(\boldsymbol{\ell}).
 \label{eq:bf-vector-fitting-mass}
\end{equation}}
For a positive target whose feasible mass is at least {$\gamma_t$}, scan the levels in decreasing order and choose a positive-mass boundary {$\tau_t(e)$} satisfying}
{\begin{equation}
 \sum_{m>\tau_t(e)}\pi_m^{e,t}
 \le\gamma_t\le
 \sum_{m\ge\tau_t(e)}\pi_m^{e,t}.
 \label{eq:bf-threshold}
\end{equation}}
{If the realized vector state lies in {$\mathcal S(e)$}, accept when {$M_{t-1}>\tau_t(e)$}; at equality, accept with probability}
{\begin{equation}
 \beta_t(e)=
 \frac{\gamma_t-\sum_{m>\tau_t(e)}\pi_m^{e,t}}
      {\pi_{\tau_t(e)}^{e,t}}.
 \label{eq:bf-boundary}
\end{equation}}
Reject from all other states.

\wh{This acceptance rule prioritizes vector-feasible states with larger maximum coordinate load, extending the threshold-based selection rule of \citet{jiangMaZhang2024}.}
\wh{In multiple dimensions, states with the same maximum load can differ in whether a proposal fits and in their maximum load after acceptance.}
\wh{The algorithm therefore retains the joint distribution of actual vector loads to calibrate acceptance and track state transitions while preserving resource complementarity.}
\wh{This use of vector states within the best-fit selection rule motivates the name \emph{Vector-Aware Best-Fit}.}

\wh{To update this distribution, let} {$a_t(\boldsymbol{\ell},e)$} be the acceptance probability prescribed by this rule when {$\mathbf L_{t-1}=\boldsymbol{\ell}$} and edge {$e$} is proposed at epoch {$t$}.
Put {$q_{0,t}=1-\sum_e q_{e,t}$}.
The forward recursion is
{\begin{align}
 \nu_t(\boldsymbol{\ell})
 ={}&\nu_{t-1}(\boldsymbol{\ell})
 \left[q_{0,t}+\sum_e q_{e,t}
       \bigl(1-a_t(\boldsymbol{\ell},e)\bigr)\right]\notag\\
 &+\sum_e q_{e,t}\nu_{t-1}(\boldsymbol{\ell}-\mathbf s_e)
               a_t(\boldsymbol{\ell}-\mathbf s_e,e).
 \label{eq:bf-vector-dp}
\end{align}}


We summarize the \wh{VABF} algorithm in Algorithm~\ref{alg:best-fit} with the offline preparation in Algorithm~\ref{alg:bf-offline}.
Algorithm~\ref{alg:bf-offline} first solves the LP relaxation, fixes {$\wh{\mathbf z}$} and its routing probabilities.
Given the deterministic acceptance targets \(\{\gamma_{u,t}\}\), offline preparation computes the state distributions and acceptance thresholds. The subsequent analysis establishes sufficient conditions on these input parameters for feasibility and performance guarantees.

\begin{algorithm}[!htbp]
\caption{Offline Preparation}
\label{alg:bf-offline}
\begin{algorithmic}[1]{\Require Probabilities {$p_{vt}$}, rewards {$w_{uv}$}, and consumption vectors {$\mathbf s_{uv}$}; targets {$\gamma_{u,t}$}.
\Ensure LP solution {$\wh{\mathbf z}$}, proposal probabilities {$q_{e,t}$}, thresholds {$\tau_{u,t}(e)$}, and boundary probabilities {$\beta_{u,t}(e)$}.
\State Solve LP \eqref{eq:planning-lp} and fix an optimal solution {$\wh{\mathbf z}$}, independently of the targets.
\State Compute {$q_{e,t}$} by \eqref{eq:routed-probability}.
\State Initialize {$\nu_{u,0}(\mathbf0)=1$} for every server {$u$}.
\For{{$t=1,\ldots,T$} and each server {$u$}}
 \State For {$\gamma_{u,t}>0$}, compute {$\tau_{u,t}(e)$} and {$\beta_{u,t}(e)$} for each positive-probability proposal by \eqref{eq:bf-vector-fitting-mass}--\eqref{eq:bf-boundary}.
 \State Update {$\nu_{u,t}$} by \eqref{eq:bf-vector-dp}.
\EndFor}
\end{algorithmic}
\end{algorithm}

\begin{algorithm}[!htbp]
\caption{Vector-Aware Best-Fit (VABF)}
\label{alg:best-fit}
\begin{algorithmic}[1]{\Require Probabilities {$p_{vt}$}, rewards {$w_{uv}$}, and consumption vectors {$\mathbf s_{uv}$}; targets {$\gamma_{u,t}$}.
\State Run Algorithm~\ref{alg:bf-offline} to obtain {$\wh{\mathbf z}$}, {$q_{e,t}$}, 
{$\tau_{u,t}(e)$}, and {$\beta_{u,t}(e)$}.
\State Initialize {$\mathbf L_u=\mathbf0$} for every server {$u$}.
\For{{$t=1,\ldots,T$}}
 \State Observe the type {$v=V_t$} of request {$t$}.
 \State Sample at most one candidate server {$u$} with probabilities {$z_{uv}/\mu_v$}.
 \If{a candidate server {$u$} is drawn}
  \State Set {$e=(u,v)$} and {$M_u=\|\mathbf L_u\|_\infty$}.
  \If{{$\gamma_{u,t}>0$} and {$\mathbf L_u+\mathbf s_e\le\mathbf1$}}
   \If{{$M_u>\tau_{u,t}(e)$}, or {$M_u=\tau_{u,t}(e)$} and a fresh Bernoulli{$(\beta_{u,t}(e))$} coin succeeds}
    \State Assign to {$u$} and update {$\mathbf L_u\gets\mathbf L_u+\mathbf s_e$}.
   \EndIf
  \EndIf
 \EndIf
 \State Reject request {$t$} if it was not assigned.
\EndFor}
\end{algorithmic}
\end{algorithm}

\subsection{{Competitive Ratio and Analysis}}
\label{sec:main-guarantee}
\label{sec:core-analysis}
The following theorem gives a guarantee that depends only on the maximum resource dimension.
It uses the same acceptance target at every server and epoch.
Let $\ALG_{\wh{\mathrm{VABF}}}$ denote the total reward realized by \wh{VABF}.

\begin{samepage}
\begin{theorem}[{Dimension-only guarantee}]
\label{thm:discrete-main}
\wh{VABF} (Algorithm~\ref{alg:best-fit}) with constant targets {$\gamma_{u,t}=R(D)$}
\hh{achieves} expected reward
{\begin{equation}
 \E[\ALG_{\wh{\mathrm{VABF}}}]=R(D)Z_{\mathrm{LP}}\ge R(D)\OPT.
 \label{eq:adaptive-reward}
\end{equation}}
Here we define 
{\begin{equation}
 R(D)=\max\left\{\frac1{2D+1+e^{-2D}},\frac4{7D+8.056}\right\}.
 \label{eq:dimension-ratio}
\end{equation}}
\end{theorem}
\end{samepage}

\paragraph{Analysis plan.}
The routing rule reduces the analysis to one server because each server receives a proposal stream that is independent across epochs.
\wh{For this server, feasibility means attaining the prescribed acceptance target while respecting every coordinate capacity.}
\wh{Since every proposal fits at zero load, it suffices that the probability of this empty state before an epoch is at least that epoch's target.}
\wh{For feasible epochs, Lemmas~\ref{lem:bf-marginal} and~\ref{lem:bf-jmz} link the targets to this probability through expected loads and heavy counts.}
\wh{Proposition~\ref{lem:bf-targets} establishes feasible constant targets by induction from zero load; Lemma~\ref{lem:network-reward} extends the guarantee to the network.}
\wh{Separately, Lemmas~\ref{lem:local-load} and~\ref{lem:gamma-monotonicity} certify $R(D)$ as a feasible common target.}
\wh{Applying Lemma~\ref{lem:network-reward} at this target and using the LP upper bound in Lemma~\ref{lem:planning-upper} proves Theorem~\ref{thm:discrete-main}.}
Figure~\ref{fig:main-proof-flow} \wh{summarizes the proof sketch.}

\begin{figure}[htbp]
\centering
\includegraphics[width=\textwidth]{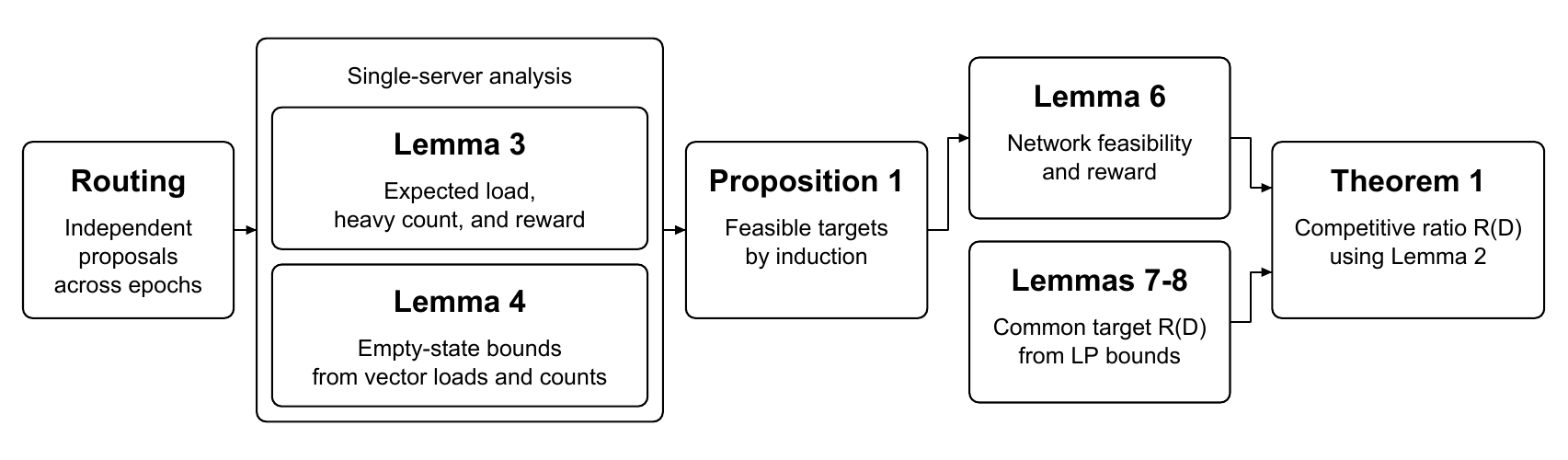}
\caption{\wh{Main steps in the proof of Theorem~\ref{thm:discrete-main}.}}
\label{fig:main-proof-flow}
\end{figure}

\paragraph{Single-server stream.}
\label{sec:single-server}
Fix a server {$u$} and suppress its subscript through \wh{Proposition}~\ref{lem:bf-targets}.
The {$t$}-th arrival produces an independent proposal with mutually exclusive outcomes and probabilities {$q_{e,t}$} from \eqref{eq:routed-probability}.
A request routed elsewhere or not routed produces a null outcome for this server.
A proposal on edge {$e$} has the fixed scalar size {$b_e\in(0,1]$} and reward {$w_e\ge0$} specified by the model.
Define
{\[
 \psi_t=\sum_e q_{e,t}b_e,\qquad
 h_t=\sum_{e:b_e>1/2}q_{e,t},\qquad
 Z_t=\sum_e q_{e,t}w_e.
\]}
Write {$B=\sum_t\psi_t$}, {$\theta=\sum_t h_t$}, and {$Z=\sum_t Z_t$} for this server's \emph{LP-induced scalar load}, \emph{expected heavy-proposal count}, and LP value.
For analysis, let {$X_t$} be the \emph{cumulative scalar load}, the sum of {$b_e$} over accepted assignments, with {$X_0=0$}.
Let {$J_t$} be the accepted heavy-assignment count.
By definition, we have 
\begin{equation}\label{eq:xmmonotone}
    M_t\le X_t, \quad X_t-X_{t-1}\ge M_t-M_{t-1}.
\end{equation}

{\paragraph{Feasibility at an epoch.}
Fix deterministic targets {$\gamma_1,\ldots,\gamma_T\in[0,1]$} for this server.
Starting from {$\mathbf L_0=\mathbf0$}, suppose \wh{VABF} has implemented the prescribed targets at epochs {$1,\ldots,i-1$} while respecting all capacities.
We call epoch {$i$} \emph{feasible} for target {$\gamma_i$} if and only if every edge {$e$} with {$q_{e,i}>0$} satisfies
{\begin{equation}
 \Prob(\mathbf L_{i-1}+\mathbf s_e\le\mathbf1)\ge\gamma_i.
 \label{eq:bf-fitting-condition}
\end{equation}}
The probability averages over earlier arrivals and policy randomization, and equals the total mass of vector-feasible states for that proposal.
For {$\gamma_i>0$}, this condition ensures that the threshold and boundary probability in \eqref{eq:bf-threshold}--\eqref{eq:bf-boundary} are well-defined and implement {$\gamma_i$}.
Thus feasibility at an epoch requires that the prescribed acceptance target can be attained while respecting all capacity constraints.}

\begin{lemma}
\label{lem:bf-marginal}
{Let {$\gamma_1,\ldots,\gamma_t\in[0,1]$} be deterministic targets, and suppose epochs {$1,\ldots,t$} are feasible for these targets.
Then, for each {$i=1,\ldots,t$}, {$\mathbf L_i\le\mathbf1$} on every realization, and each edge {$e$} with {$q_{e,i}>0$} is accepted with probability {$\gamma_i$} conditional on its being proposed at epoch {$i$}.
Moreover,
{\begin{align}
 \E X_t&=\sum_{i\le t}\gamma_i\psi_i,&
 \E J_t&=\sum_{i\le t}\gamma_i h_i,&
 \E[\text{reward through time }t]&=\sum_{i\le t}\gamma_i Z_i .
 \label{eq:bf-marginals}
\end{align}}}
\end{lemma}{\emph{Proof idea.}
The current proposal is independent of the pre-request vector load.
\wh{VABF} selects exactly the prescribed unconditional mass among vector-feasible states, using its boundary coin if needed.
Every accepted proposal respects all coordinate capacities.
Appendix~\ref{sec:marginal} gives the complete proof.}

A common sufficient condition for \eqref{eq:bf-fitting-condition} is
\begin{equation}
 \Prob(M_{t-1}=0)\ge\gamma_t,
 \label{eq:bf-empty-sufficient}
\end{equation}
because the zero vector fits every legal consumption vector.
Once epochs {$1,\ldots,t-1$} are feasible, establishing this bound makes epoch {$t$} feasible and closes induction from the empty initial state.

{\paragraph{Empty-state bounds.}
\label{sec:target-certificates}
For nonincreasing targets, the following lemma bounds empty-state probability $\Prob(M_t=0)$ using maximum coordinate load {$M_t$}, cumulative scalar load {$X_t$}, and the accepted heavy-assignment count {$J_t$}.
}

\begin{lemma}[{Empty-state bounds}]
\label{lem:bf-jmz}
{Suppose epochs {$1,\ldots,t$} are feasible for targets {$0\le\gamma_t\le\cdots\le\gamma_1\le c\le1$}, where {$c>0$}, and let
{\begin{equation}
 \Psi_t=\min\{X_t-M_t,\,(X_t-2/3)^+\}.
 \label{eq:bf-complementarity}
\end{equation}}
Then
{\begin{equation}
 \Prob(M_t=0)\ge H_c(\E M_t),\qquad
 H_c(y)=\min\{1-c-y,\ 1-2y-ce^{-2y/c}\},
 \label{eq:bf-H-definition}
\end{equation}}
and
{\begin{equation}
 \Prob(M_t=0)\ge
 1-\left(\tfrac32\E [X_t-\Psi_t]+\tfrac14\E J_t+1.014c\right).
 \label{eq:bf-refined-empty}
\end{equation}}}
\end{lemma}

\zz{
The two bounds in Lemma \ref{lem:bf-jmz} provide complementary ways to verify the sufficient feasibility condition (\ref{eq:bf-empty-sufficient}). Inequality \eqref{eq:bf-H-definition} extends the utilization analysis of \cite{jiangMaZhang2024} to maximum coordinate load. Inequality \eqref{eq:bf-refined-empty} additionally accounts for cumulative scalar load and accepted heavy-assignment counts, with an explicit compensation term for resource complementarity.} 

\wh{To establish inequality \eqref{eq:bf-refined-empty}, we construct a pathwise upper bound on the nonempty-state indicator $\one_{\{M_t>0\}}$ using quantities whose expectations can be controlled under VABF: the cumulative scalar load {$X_t$}, maximum coordinate load {$M_t$}, the accepted heavy-assignment count {$J_t$}, and the complementarity compensation {$\Psi_t$}. Taking expectations then bounds the nonempty-state probability and, consequently, the empty-state probability needed to certify VABF’s acceptance targets.
We start with}
\begin{equation}
 \wh{A_t=\frac32(X_t-\Psi_t)+\frac14J_t.}
 \label{eq:bf-indicator-base}
\end{equation}
\wh{
By the definition of $\Psi_t$, $X_t-\Psi_t
=\max\left\{M_t,\min(X_t,2/3)\right\}$. Thus, if $X_t\ge2/3$, then $X_t-\Psi_t\ge2/3$, leading to $A_t\ge1$.
If $J_t\ge1$, an accepted heavy assignment implies $M_t>1/2$; hence 
$A_t\geq \frac32M_t+\frac14J_t\ge1$.
The heavy-count term therefore supplies the additional weight needed to cover states containing a heavy assignment while keeping the load coefficient at $3/2$.
}

\zz{The compensation term captures a distinctly multidimensional feature. $X_t$ sums the largest coordinate consumption of each accepted assignment, $M_t$ measures the largest coordinate of their cumulative vector load. Assignments concentrated in different coordinates can therefore yield $X_t>M_t$. Subtracting $\Psi_t$ credits part of this resource complementarity without losing indicator coverage in the two cases above. In one dimension, $X_t=M_t$ and $\Psi_t=0$.}

\wh{
For the remaining nonempty states, the base variable $A_t$ may be less than one. These states are covered by three regions: $E_1=\{0<M_t\le1/3\}$, $E_2=\{1/3<M_t\le1/2\}$, and $E_3=\{J_t=0,\ 1/2<M_t\le X_t\le2/3\}$. On these regions, 
$A_t$ is bounded below by $0$, $1/2$, and $3/4$, respectively.
Adding the corresponding deficit weights gives the pathwise inequality}
\[
 \wh{\widehat A_t=A_t+\one_{E_1}+\frac12\one_{E_2}+\frac14\one_{E_3}\ge\one_{\{M_t>0\}}.}
\]
Taking expectations yields
\[
 \wh{\Prob(M_t>0)\le\E\widehat A_t=\E A_t+\Prob(E_1)+\frac12\Prob(E_2)+\frac14\Prob(E_3).}
\]
\zz{This construction separates the multidimensional geometry from the algorithm-dependent analysis. The pathwise inequality credits complementary vector loads through $\Psi_t$.}

\wh{The three region probabilities are coupled through the state transitions induced by VABF. 
For example, when a proposal has scalar size $1/3<b\le1/2$, a newly entering state in $E_3$ must have $X_{t-1}\le2/3-b<1/3$ before acceptance.
The preceding state cannot be empty, since acceptance from an empty state yields maximum load $b\le1/2$,  it must therefore belong to  $E_1$.
These transition restrictions, together with VABF’s priority for vector-feasible states with higher maximum load, control probability accumulation across the three regions. 
Under the conditions of Lemma~\ref{lem:bf-jmz}, the weighted probability sum satisfies $\Prob(E_1)+\frac12\Prob(E_2)+\frac14\Prob(E_3)\le1.014c$.
Substituting \eqref{eq:bf-indicator-base} into the expectation inequality and subtracting from one gives \eqref{eq:bf-refined-empty}.}
The complete proof is given in Appendix~\ref{sec:load-refinement}.

\zz{Combining these bounds with the expected-load identities in Lemma \ref{lem:bf-marginal} yields the admissible constant targets, which will be established in \wh{Proposition} \ref{lem:bf-targets}}. 

{\paragraph{From empty-state bounds to constant targets.}
\wh{We seek a constant target $\gamma$ such that $\Prob(M_t=0)\ge\gamma$ whenever epochs $1,\ldots,t$ are feasible, for $1\le t<T$.
The empty initial state makes the first epoch feasible, and \eqref{eq:bf-empty-sufficient} then ensures feasibility at epoch $t+1$, completing the induction.} If we set constant targets $\gamma_t = c=\gamma$, 
when epochs {$1,\ldots,t$} are feasible, Lemma~\ref{lem:bf-marginal} gives
$
 \E X_t=\gamma \sum_{i\le t}\psi_i\le \gamma B
 $
 and
 $
 \E J_t=\gamma\sum_{i\le t}h_i\le \gamma \theta.
$
Since {$M_t\le X_t-\Psi_t$} and {$H_\gamma$} is nonincreasing (by definitions in Lemma~\ref{lem:bf-jmz}), \eqref{eq:bf-H-definition} with {$c=\gamma$} gives 
$\Prob(M_t=0)\ge H_\gamma(\E M_t)\ge H_\gamma(\gamma B-\E\Psi_t).$
Together with the compensation term in \eqref{eq:bf-refined-empty}, this establishes that a lower bound on {$\E\Psi_t$} improves both empty-state bounds.}

To obtain such a lower bound, we exploit the pathwise monotonicity of \wh{$\Psi_t$}.
Inequality \eqref{eq:xmmonotone} and the nondecreasing cumulative scalar load $X_t$ imply, through \eqref{eq:bf-complementarity}, that $\Psi_t\ge \Psi_2$ where $\Psi_2$ is the compensation achieved by an ordered pair of proposals in the first two epochs. \zz{Thus, complementarity generated by accepted proposals in the first two epochs provides a persistent contribution to the compensation term. If these proposals are} $(e, f)$, satisfy $2/3<b_e+b_f\le1$, and are both accepted, then under this sample path, $X_2=b_e+b_f$, $M_2=\|\mathbf s_e+\mathbf s_f\|_\infty$, and $\Psi_2 = \min\left\{
 b_e+b_f-\|\mathbf s_e+\mathbf s_f\|_\infty,\,
 b_e+b_f-\tfrac23
 \right\}$.
We use this pairwise contribution to define $\varepsilon$ below and formally prove a lower bound on $\E\Psi_t$ in terms of $\varepsilon$  in the following Lemma~\ref{lem:pair-credit} (the complete proof is in Appendix~\ref{sec:pair-credit-proof}).
For $T\ge 2$, define the pairwise complementarity credit {\begin{equation}
 \varepsilon=
 \max_{\substack{e,f\\2/3<b_e+b_f\le1}}
 q_{e,1}q_{f,2}
 \min\left\{
 b_e+b_f-\|\mathbf s_e+\mathbf s_f\|_\infty,\,
 b_e+b_f-\tfrac23
 \right\},
 \label{eq:bf-pair-credit}
\end{equation}}
where the maximum is over eligible ordered pairs of edges incident to the server.
Set $\varepsilon=0$ if $T<2$ or no eligible pair exists. \zz{Each pair’s compensation is weighted by its proposal probability in the first two epochs.}  

\begin{lemma}
\label{lem:pair-credit}
\wh{For a constant target $\gamma_t=\gamma\in(0,1]$, if epochs $1,\ldots,t$ are feasible for some $2\le t\le T$, then $\E\Psi_t\ge\gamma\varepsilon$.}
\end{lemma}


We now choose constant $\gamma$ so that these bounds certify the next epoch.
It suffices to ensure that $\Prob(M_t=0)\ge\gamma$.
For $x=B-\varepsilon$ and $y=\theta$, substituting the expectation bounds of $\E[X_t-\Psi_t]\le \gamma x$ and $\E J_t\le \gamma y$ into \eqref{eq:bf-H-definition} and \eqref{eq:bf-refined-empty} with $c=\gamma$ gives the sufficient conditions $1-\gamma-\gamma x\ge \gamma$ 
and 
$1-2\gamma x - \gamma e^{-2x}\ge \gamma$
for the first inequality, 
or 
$1-(\frac32 \gamma x+\frac14 \gamma y+1.014\gamma)\ge \gamma$
for the second.
These conditions define the feasible range for $\gamma$: $\gamma\le \Gamma_1(x)=1/[1+\max\{1+x,\,2x+e^{-2x}\}]$ or $\gamma\le \Gamma_2(x,y)=4/(6x+y+8.056)$, and we set
{\begin{equation}
 \Gamma(x,y)=\max\{\wh{\Gamma_1(x)},\,\wh{\Gamma_2(x,y)}\},
 \qquad x,y\ge0.
 \label{eq:adaptive-target}
\end{equation}}
Then the sufficient condition for feasible targets is that constant targets $\gamma\le \Gamma(x,y)$.
The following \wh{proposition} applies these conditions inductively from the empty initial state and gives the resulting reward guarantee. The complete proof is in Appendix~\ref{sec:target-proof}.

\begin{proposition}[{Feasible constant targets}]
\label{lem:bf-targets}
Every constant target 
$0\le\gamma_{t} =\gamma \le\Gamma(B-\varepsilon,\theta)$
is feasible for \wh{VABF} and gives expected total reward {$\gamma Z$}.
\end{proposition}

{\paragraph{From a single server to the network.}
Restore the server index, so that {$B$}, {$\theta$}, {$\varepsilon$}, and {$Z$} become {$B_u$}, {$\theta_u$}, {$\varepsilon_u$}, and {$Z_u$}.
Equivalently,
{\begin{equation}
B_u=\sum_{e\in\delta(u)}b_ez_e,\quad
\theta_u=\sum_{\substack{e\in\delta(u)\\b_e>1/2}}z_e,\quad
Z_u=\sum_{e\in\delta(u)}w_ez_e,
\label{eq:local-loads}
\end{equation}}
and {$\sum_u Z_u=Z_{\mathrm{LP}}$}.}

\begin{lemma} \label{lem:network-reward}
\wh{Every choice of targets $0\le\gamma_{u,t}=\gamma_u\le\Gamma(B_u-\varepsilon_u,\theta_u)$ for all $u,t$ is feasible for VABF and gives expected total reward $\sum_u\gamma_u Z_u$.
In particular, choosing $\gamma_{u,t}=\Gamma(B_u-\varepsilon_u,\theta_u)$ for all $u,t$ guarantees a competitive ratio of at least $\min_u\Gamma(B_u-\varepsilon_u,\theta_u)$, with}
\begin{equation}
\label{eq:bf-cr}
\displaystyle
 \E[\ALG_{\wh{\mathrm{VABF}}}]
 =\sum_u\Gamma(B_u-\varepsilon_u,\theta_u)Z_u
 \ge\min_u\Gamma(B_u-\varepsilon_u,\theta_u)  Z_{\mathrm{LP}}
 \ge\min_u\Gamma(B_u-\varepsilon_u,\theta_u)\OPT.
\end{equation}
\end{lemma}
\emph{Proof idea.} Proposition~\ref{lem:bf-targets} gives feasibility and expected reward $\wh{\gamma_u Z_u}$ at each server.
Routing proposes each request to at most one server, so these policies form a feasible network allocation.
\wh{Summing over servers gives the reward identity; choosing the stated server-specific targets and applying} Lemma~\ref{lem:planning-upper} \wh{gives the competitive-ratio guarantee.}
The complete proof is in Appendix~\ref{sec:network-reward-proof}.

The next lemma bounds the loads $B_u$ and $\theta_u$ by {$D$} using the capacity and heavy-assignment count constraints at {$u$}.
\zz{The first bound follows by summing the coordinatewise capacity constraints; the second follows by summing the heavy-assignment count constraints, since every heavy edge is heavy in at least one coordinate. These bounds convert the instance-specific acceptance guarantee into a dimension-only guarantee.}


\begin{lemma}[{LP load bounds}]
\label{lem:local-load}
For every server, we have $B_u\le D$ and $\theta_u\le D$.
\end{lemma}
{\emph{Proof idea.}
Each edge's maximum coordinate is bounded by the sum of its coordinates; summing the capacity rows bounds {$B_u$}.
Every heavy edge is heavy in at least one coordinate; summing the heavy-assignment count rows bounds {$\theta_u$}.
Both sums have at most {$D$} rows.
Appendix~\ref{sec:local-load-proof} gives the complete argument.}

\wh{The two bounds in 
Lemma~\ref{lem:local-load} \hh{connect} multidimensional capacity to the acceptance analysis. 
Both bounds are tight. 
Consider a single server with $D$ unit-capacity dimensions and $D$ request types, each arriving deterministically once, with unit reward and consumption of one unit in its own coordinate and zero elsewhere.
The optimal LP assigns $z_{uv}=1$ to every type, giving $B_u=\theta_u=D$. 
These assignments are jointly vector-feasible despite their total scalar size being $D$, illustrating how resource complementarity allows both loads to grow with the dimension.
Thus, neither universal upper bound can be reduced below $D$ without additional restrictions on the instance.
Theorem~\ref{thm:limit-heavy} \hh{shows that no online policy can guarantee a competitive ratio greater than $1/(D+1)$, even on all-heavy instances.}}

To apply the load bounds to our analysis, we introduce the following lemma to establish the monotonicity of $\Gamma$. The complete proof is in Appendix~\ref{sec:gamma-monotonicity-proof}.
\begin{lemma}
\label{lem:gamma-monotonicity}
The function $\Gamma$ defined by \eqref{eq:adaptive-target} is strictly decreasing in $x$ for each fixed $y\ge0$ and nonincreasing in $y$ for each fixed $x\ge0$.
\end{lemma}

Now we are ready to give the proof of Theorem~\ref{thm:discrete-main}:
By Lemma~\ref{lem:local-load} \wh{and Lemma~\ref{lem:gamma-monotonicity}}, {$\Gamma(B_u-\varepsilon_u,\theta_u)\ge \Gamma(B_u,\theta_u)\ge\Gamma(D,D)=R(D)$}.
Applying Lemma~\ref{lem:network-reward} with $\gamma_{u,t}=R(D)$ for all $u,t$ gives a feasible network allocation and $\E[\ALG_{\wh{\mathrm{VABF}}}]=R(D)\sum_u Z_u=R(D)Z_{\mathrm{LP}}$.
Applying \hyperref[lem:planning-upper]{Lemma~\ref*{lem:planning-upper}} proves \eqref{eq:adaptive-reward}.
Appendix~\ref{sec:network-proof} supplies the complete proof.

\emph{Revisit the empty-state bounds in Lemma~\ref{lem:bf-jmz}.} 
\zz{The two branches of
\[
R(D) = \max\left\{\frac1{2D+1+e^{-2D}},\frac4{7D+8.056}\right\}
\]
reflect complementary approaches to controlling empty-state probability. 
The maximum-load bound \eqref{eq:bf-H-definition} extends the utilization analysis of \citet{jiangMaZhang2024} and yields the first branch after applying the LP capacity constraints. Our refined bound \eqref{eq:bf-refined-empty} additionally incorporates the accepted heavy-assignment count $J_t$ and the complementarity compensation $\Psi_t$. Since a heavy assignment implies $M_t>1/2$, the count term supplies part of the weight needed to cover the nonempty-state indicator, reducing the load coefficient from $2$ to $3/2$. The LP capacity and heavy-count constraints then give a combined dimension coefficient of $3/2+1/4=7/4$, yielding
$1/((7/4)D+2.014)=4/(7D+8.056)$.
The smaller dimension coefficient offsets the larger constant term for every $D\ge5$, where the refined bound provides a strictly stronger dimension-only guarantee. For $D=1,\ldots,4$, the maximum-load bound remains stronger.}

{\paragraph{Optional refinement.}
\hh{Theorem~\ref{thm:discrete-main} uses the common target $\gamma_{u,t}=R(D)$ to obtain a dimension-only guarantee.}
\hh{We next use instance-specific load information to refine this guarantee.}
The server-specific targets $\gamma_{u,t}=\Gamma(B_u-\varepsilon_u,\theta_u)$ in Lemma~\ref{lem:network-reward} give at least $R(D)Z_{\mathrm{LP}}$ in expected reward, since $\Gamma(B_u-\varepsilon_u,\theta_u)\ge R(D)$ for every server.
If {$Z_{\mathrm{LP}}>0$} and some server has {$\varepsilon_u>0$} and {$Z_u>0$}, then \wh{Lemma~\ref{lem:gamma-monotonicity} gives a strict increase in its target, and} the server-specific choice satisfies
{\begin{align*}
 \frac{\E[\ALG_{\wh{\mathrm{VABF}}}]}{Z_{\mathrm{LP}}}
 &\ge R(D)+\frac{Z_u}{Z_{\mathrm{LP}}}
 \bigl[\Gamma(B_u-\varepsilon_u,\theta_u)
       -\Gamma(B_u,\theta_u)\bigr]\\
 &>R(D).
\end{align*}}
For example, {$\varepsilon_u>0$} if two edges have their scalar sizes in {$(1/3,1/2]$}, have different unique maximizing coordinates, and satisfy {$q_{e,1}q_{f,2}>0$}.
The reward gain depends on {$B_u$}, {$\theta_u$}, {$\varepsilon_u$}, and each server's share of LP reward.
}

The state-distribution recursion and an implementation are detailed in Appendix~\ref{app:implementation-rounding}.
\paragraph{Algorithmic and analytical distinctions from \citet{jiangMaZhang2024}.}
Now we discuss how our approach differs from that of \citet{jiangMaZhang2024} in the ex ante LP, the acceptance rule, and the feasibility analysis.
On the algorithmic side, we augment the coordinatewise capacity constraints with heavy-assignment count constraints.
These inequalities strengthen the ex ante relaxation and provide the bound $\theta_u\le D$ in Lemma~\ref{lem:local-load}, linking the LP to the subsequent acceptance analysis.
Following candidate routing, the acceptance rule maintains the joint distribution of actual vector loads and checks whether the proposed consumption vector fits in every coordinate.
This joint state representation preserves resource complementarity, while the maximum coordinate load orders the feasible states for calibrated acceptance.
On the analytical side, Lemma~\ref{lem:bf-jmz} establishes two complementary lower bounds on the empty-state probability.
Inequality \eqref{eq:bf-H-definition} extends the utilization argument of \citet{jiangMaZhang2024} to vector loads.
Inequality \eqref{eq:bf-refined-empty} combines \hh{$\E[X_t-\Psi_t]$ and $\E J_t$, with $\Psi_t$ accounting for resource complementarity.}
The heavy-count constraints in the LP control the expected count entering this second bound and thereby support its acceptance guarantee.
Together, the two bounds yield the two branches of $R(D)$ in Theorem~\ref{thm:discrete-main}: the branch $1/(2D+1+e^{-2D})$ is larger for $D=1,\ldots,4$, whereas the other branch $4/(7D+8.056)$ is larger for every $D\ge5$.

{\section{Dimensional Hardness and All-Heavy Optimality}
\label{sec:limits}}
The guarantee improves on scalarization but still decreases with the resource dimension.
We now show that inverse-dimensional dependence is intrinsic to the allocation problem.
On the all-heavy subclass, the optimal uniform ratio in our nonstationary model is exactly {$1/(D+1)$}, even for policies that fully exploit the original vector capacities.
This also bounds the general nonstationary model.
A separate construction gives another upper bound relative to the LP relaxation.
\subsection{All-Heavy Optimality and Hardness}
Recall that an instance is \emph{all-heavy} when every feasible edge satisfies {$b_e>1/2$}, as defined in Section~\ref{sec:model}.
\wh{VABF} may accept several heavy proposals when their consumption vectors fit together.

This subclass isolates the conflict between accepting an early request and preserving capacity for later arrivals.
The next theorem gives a specialized constant target and a matching upper bound, thereby identifying a dimensional barrier that also applies to the general model.

\begin{theorem}[{All-heavy optimality}]
\label{thm:limit-heavy}
On all-heavy instances, constant targets {\mbox{$\gamma_u^{\mathrm H}=1/(1+\theta_u)$}} give \wh{VABF} expected reward
{\[
\E[\ALG_{\wh{\mathrm{VABF}}}]
=\sum_u\frac{Z_u}{1+\theta_u}
\ge\frac{\OPT}{D+1}.
\]}
For every integer {$D\ge1$} and {$\zeta>0$}, there is a finite, single-server all-heavy instance such that every online policy $\Pi$ satisfies
{\[
\frac{\E[\ALG_\Pi]}{\OPT}
\le\frac{1}{D+1}+\zeta.
\]}
Hence, the optimal competitive ratio on this subclass is {\mbox{$1/(D+1)$}}.
\end{theorem}

{\emph{Proof idea (competitive ratio).}
On all-heavy instances, {$M_t>0$} is equivalent to {$J_t\ge1$}, so the accepted heavy-assignment count alone gives {$\Prob(M_t=0)\ge1-\E J_t$}.
Using Lemma~\ref{lem:bf-marginal}, the target {$\gamma_u^{\mathrm H}=1/(1+\theta_u)$} keeps this empty-state probability at least the target by induction.
The reward identity and {$\theta_u\le D$} then yield the competitive ratio.}

\zihao{{\emph{Proof idea (hardness).}
For $D\ge2$ and $\zeta>0$, choose
$0<p<\min\{1/(2D),\zeta/2\}$ and an integer $n\ge1$
such that $(1-p^D)^n\le p$.
We construct a single-server instance with $n$ groups
of $D$ request types, each arriving independently
with probability $p$ at a designated epoch and
offering unit reward.
Following the compatible-group construction of
\citet{naoriRaz2019}, the consumption vectors are
chosen so that all requests within a group fit together,
whereas any two requests from different groups
are incompatible.
All these requests are heavy.
At the final epoch, a request consuming the entire
capacity vector arrives with probability $p$ and
offers reward $1/p$. \\
Once an online policy accepts an early request,
it can subsequently accept only requests from
the same group, yielding expected total reward
at most $1+(D-1)p$.
If no early request is accepted, the expected reward
from the final request is at most one.
In contrast, the offline optimum can accept the
final request whenever it arrives and otherwise
serve all $D$ requests from any complete group,
where a group is complete if all its $D$ requests arrive.
Since each group is complete with probability $p^D$,
independently across groups, we obtain
\[
\frac{\E[\ALG_\Pi]}{\OPT}
\le
\frac{1+(D-1)p}
{1+(1-p)D[1-(1-p^D)^n]}
\le
\frac{1+(D-1)p}{1+D(1-p)^2}
<
\frac{1}{D+1}+2p
<
\frac{1}{D+1}+\zeta.
\]
For $D=1$, consider a two-epoch instance with
a certain unit-reward request followed by a
unit-size request of reward $1/p$ that arrives
with probability $0<p<\min\{1/2,\zeta/2\}$.
The optimal online reward is one, whereas
$\OPT=2-p$, yielding
\[
\frac{\E[\ALG_\Pi]}{\OPT}
\le \frac{1}{2-p}
<\frac12+\zeta.
\]
Appendix~\ref{app:hardness-details} provides
the complete construction and proof.
}}

\begin{corollary}[Dimensional upper bound]
\label{cor:dimensional-upper}
For every integer {$D\ge1$}, no online policy can guarantee a competitive ratio greater than {$1/(D+1)$}.
\end{corollary}
This follows directly from Theorem~\ref{thm:limit-heavy}, since all-heavy instances form a subclass of the general model.

\subsection{Upper Bound Relative to the LP Relaxation}
\label{sec:lp-hardness}
The preceding upper bound uses the expected offline optimum as its benchmark.
\hh{Relative to the strengthened LP value $Z_{\mathrm{LP}}$,} Theorem~\ref{thm:discrete-main} \hh{guarantees expected reward at least $R(D)Z_{\mathrm{LP}}$.}
\hh{The next result shows that no online policy can guarantee a ratio greater than $1/(D+2)$ relative to this LP, uniformly over instances with $D$ resource dimensions.}

\begin{theorem}[{Upper bound relative to the LP}]
\label{thm:lp-online-limit}
For every integer {$D\ge1$} and {$\zeta>0$}, there is an instance such that every online policy $\Pi$ satisfies
{\[
 \frac{\E[\ALG_\Pi]}{Z_{\mathrm{LP}}}
 \le \frac1{D+2}+\zeta.
\]}
\end{theorem}

\emph{Proof idea.}
For a prescribed $\zeta>0$, take one server with $D$ dimensions and $T=D+2$ epochs, and choose $0<\omega<\min\{1/[2(D+1)],\zeta/D\}$ and $p_\omega=1-2(D+1)\omega$.
\begin{itemize}
\item At $t=1$, a request with vector $(1/2,\omega,\ldots,\omega)$ arrives certainly; at $t=2$, a request with vector \mbox{$(1/2+\omega,\omega,\ldots,\omega)$} arrives with probability $p_\omega$. Both have reward one.
\item At $t=d+1$, for each $d=2,\ldots,D$, a request consuming one unit in coordinate $d$ and $\omega$ elsewhere arrives with probability $p_\omega$ and has reward one.
\item At $t=T$, a request with vector $\mathbf1$ and reward $1/\omega$ arrives with probability $\omega$.
\end{itemize}

Every two requests are incompatible, so a feasible policy can accept at most one.
Accepting before the final epoch earns one, and waiting for the final request has expected reward one; thus the optimal online value is one.
The LP can count the full expected supply of these mutually incompatible requests, giving $Z_{\mathrm{LP}}=2+Dp_\omega\to D+2$ as $\omega\to 0$, including one unit from the rare final request.
The chosen $\omega$ makes the online-to-LP ratio smaller than $1/(D+2)+\zeta$.
Appendix~\ref{app:lp-hardness} verifies LP feasibility and provides the complete proof.
\section{Numerical Study}
\label{sec:computations}
\wh{This section examines how routing, acceptance calibration, and recovery improve reward while retaining VABF's guarantee, and when preserving capacity for future agent sessions is most valuable.}
Section~\ref{sec:exp-design} constructs the instances, arrival laws, and performance measure.
Section~\ref{sec:exp-heuristics} \wh{introduces two guarantee-preserving enhancements of the VABF policy with server-specific targets.}
Section~\ref{sec:exp-benchmarks} compares \wh{VABF-TR} with scalar JMZ reduction, the policy of \citet{naoriRaz2019}, and greedy acceptance.

\subsection{Data and Experimental Design}
\label{sec:exp-design}
\wh{The provider decides whether to accept each arriving agent session and to which compatible service pool it should be assigned; acceptance commits its resource allowances for the planning horizon.}
We construct instances from 3,798 multi-call sessions in the frozen TraceLab v0.0.1 extract \citep{zhuEtAl2026TraceLab}, with each arrival representing \wh{an LLM-agent} session.
The four recorded resource-usage measures are input tokens including cache reads, output tokens, tool calls, and LLM calls.
\wh{We first examine variation in resource demand by dividing the frozen session sequence into five consecutive segments following its saved ordering.}
Table~\ref{tab:pom-segment-means} \wh{reports the mean usage of each resource coordinate within each segment.}

\begin{table}[!htbp]
\centering
\caption{\wh{Mean resource usage across five consecutive segments of the frozen session sequence.}}
\label{tab:pom-segment-means}
\small
\setlength{\tabcolsep}{7pt}
\begin{tabular}{@{}rrrrr@{}}
\toprule
\wh{Segment} & \shortstack{\wh{Input tokens}\\\wh{(millions)}} & \shortstack{\wh{Output tokens}\\\wh{(thousands)}} & \wh{Tool calls} & \wh{LLM calls} \\
\midrule
1 & 3.59 & 19.21 & 46.52 & 43.38 \\
2 & 12.94 & 49.80 & 133.95 & 102.05 \\
3 & 25.51 & 71.66 & 196.18 & 159.32 \\
4 & 20.51 & 64.22 & 108.39 & 98.12 \\
5 & 9.41 & 40.47 & 83.49 & 66.03 \\
\bottomrule
\end{tabular}
\par\smallskip{\footnotesize\wh{Means use all 3,798 sessions before discretization and clustering, with fractional weights for sessions crossing a segment boundary.}}
\end{table}

\wh{Mean input-token usage in segment 3 is 7.1 times that in segment 1, whereas mean output-token usage, tool calls, and LLM calls are approximately 3.7, 4.2, and 3.7 times as large, respectively.}
\wh{These differences show variation in per-session resource requirements and resource mix along the saved sequence, motivating the nonstationary laws constructed below.}

\wh{To construct admission requirements, we first discretize the four resource measures into integer reservation units.}
\wh{We aggregate the resulting profiles using resource features and session frequencies only, yielding $K=10$ or $20$ request types, including one null-arrival type.}
\wh{Each type is represented by a fixed allowance vector known at admission.}

\wh{All experiments use} four heterogeneous service pools, each with separate, nonreplenishing budgets for its retained resource coordinates.
\hh{The $D=2$ specification uses input- and output-token budgets; $D=3$ additionally includes tool-call budgets, and $D=4$ further includes LLM-call budgets.}
These specifications jointly determine the retained budget coordinates and compatible task--pool pairs.
Pool budgets stay fixed as the number of offered sessions $T$ increases.
After clustering, we assign a workload-value proxy equal to a session's raw input tokens plus \wh{$\lambda$} times its output tokens\wh{, where $\lambda$ is the output-token reward weight and $\lambda=3$ in the main study}.
A type's reward is this proxy averaged over its sessions, divided by the pooled session mean, and multiplied by a \wh{server-specific} factor.

\wh{With the catalogue fixed, we} use segment type frequencies to construct a pooled i.i.d.\ law ($S=1$) and nonstationary laws with $S=5,10$, forming five segments by pooling adjacent pairs of ten segments.
Each segment occupies $T/S$ epochs, with independent type draws from its distribution.
At each $(K,D,T)$, the three laws have the same expected type counts, budgets, rewards, and LP solutions, providing a comparison of \wh{nonstationarity at fixed aggregate demand}.

\wh{The main grid crosses $(K,D,T,S)$ to form 108 fixed, trace-calibrated instances.}
The demand levels are $T=50,100,150,200,250,300$; \wh{the main displays use $K=20$, $S=5$, and $\lambda=3$ unless varied explicitly.}
Table~\ref{tab:pom-enhancements} \wh{shows $T=100,200,300$.}
\wh{The traces supply resource profiles and type frequencies; budgets, values, independent arrivals, and advance knowledge of the laws specify the experimental model.}

We compare policies on \wh{shared arrival paths}.
\wh{Average reward (or mean reward) is the total reward averaged over the shared paths; we report it normalized by the instance's numerically computed $Z_{\mathrm{LP}}$ from the vector planning LP.}
\wh{All policies use this common denominator.}
Appendix~\ref{app:computational-details} gives the construction and all six demand levels.

\subsection{Guarantee-Preserving Heuristics}
\label{sec:exp-heuristics}
\wh{We quantify the successive reward contributions of tuning and recovery while retaining VABF's performance guarantee.}
\wh{The baseline VABF retains the original LP routing and uses the server-specific analytic targets $\Gamma(B_u-\varepsilon_u,\theta_u)$.}
\wh{For these nominal targets, Lemma~\ref{lem:network-reward} gives feasibility and an expected reward of $\sum_u\Gamma(B_u-\varepsilon_u,\theta_u)Z_u$, at least $R(D)Z_{\mathrm{LP}}$.}
\wh{Two successive enhancements preserve this guarantee: VABF with tuning (VABF-T) and VABF with tuning and recovery (VABF-TR).}
\begin{itemize}
\item \textbf{\wh{VABF-T}.} \wh{VABF's routing may leave an arrival without a server proposal, potentially leaving capacity unused.}
\wh{VABF-T constructs a candidate by increasing routing probabilities while preserving the LP's relative routing shares.}
\wh{The original analytic targets need not remain feasible under the changed proposal streams.}
\wh{We therefore keep the same single-server framework for each $u$ and seek the largest feasible constant target $\gamma_u$ for its new proposal stream.}
\wh{Before arrivals, we compare the candidate's total expected reward with that of VABF and adopt the candidate only when the comparison verifies no decrease; otherwise we retain VABF.}
\wh{Thus, VABF-T has expected reward at least as large as VABF's and preserves its performance guarantee; implementation details are given in} Appendix~\ref{app:experimental-policies}.
\item \wh{\textbf{VABF-TR.} It adds recovery to the selected VABF-T policy, which it runs as a shadow policy.}
\wh{When the shadow policy rejects a request, VABF-TR assigns it if doing so leaves enough capacity for every future assignment sequence possible from the shadow policy's current state under the arrival law and its randomized decision rule.}
\wh{The required capacity limits are computed offline; online, recovery selects the highest-reward feasible pool passing this test, breaking ties by pool index.}
\wh{Under the same arrivals and random draws, VABF-TR preserves all assignments of the shadow policy and adds only nonnegative reward, so it retains VABF-T's performance guarantee} (Appendix~\ref{app:bfr-implementation}).
\end{itemize}

Table~\ref{tab:pom-enhancements} \wh{compares the average rewards of VABF and its two enhancements.}
\wh{In the displayed instances, increasing routing probabilities and tuning feasible acceptance targets raise average reward by 29.49 to 38.83 percentage points relative to the baseline VABF.}
\wh{VABF} and \wh{VABF-T} use the same LP solution, pre-arrival load ordering, and boundary tie rule, so the comparison isolates the routing-and-target enhancement from LP selection and ordering differences.

\wh{Recovery contributes a further 5.81 to 9.03 percentage points by admitting sessions that can coexist with all protected future assignments.}
\wh{This additional reward measures the value of using capacity beyond what is needed to preserve the tuned policy's possible future allocations.}
Appendix~\ref{app:complete-experiment-results} \wh{reports the complete average rewards.}

\begin{table}[!htbp]
\centering
\caption{\wh{Average rewards of VABF and its two enhancements, $K=20$, $S=5$, and $\lambda=3$ (reward / $Z_{\mathrm{LP}}$, \%).}}
\label{tab:pom-enhancements}
\small
\setlength{\tabcolsep}{6pt}
\begin{tabular}{@{}rrrrr@{}}
\toprule
$T$ & $D$ & \wh{VABF} & \wh{VABF-T} & \wh{VABF-TR} \\
\midrule
100 & 2 & 30.49 & 64.53 & 72.25 \\
100 & 3 & 30.58 & 61.63 & 69.51 \\
100 & 4 & 28.79 & 58.28 & 67.31 \\
\addlinespace
200 & 2 & 32.46 & 71.30 & 77.89 \\
200 & 3 & 32.63 & 70.83 & 77.37 \\
200 & 4 & 32.30 & 62.89 & 71.12 \\
\addlinespace
300 & 2 & 35.36 & 74.08 & 79.88 \\
300 & 3 & 34.42 & 72.89 & 78.71 \\
300 & 4 & 32.80 & 68.67 & 77.35 \\
\bottomrule
\end{tabular}
\par\smallskip{\footnotesize \wh{Average total rewards over shared paths, normalized by the instance's $Z_{\mathrm{LP}}$ and expressed as percentages.}}
\end{table}

\subsection{Benchmark Comparisons}
\label{sec:exp-benchmarks}
\wh{We examine how VABF-TR's relative performance depends on demand, resource constraints, and arrival distributions, and when preserving capacity for future sessions is most valuable.}
\wh{The comparators comprise an immediate allocation rule and two policies from the online allocation literature.}
\begin{itemize}
\item \wh{\textbf{RJMZ.} It replaces each consumption vector by its largest coordinate, solves the resulting scalar LP, and routes with those probabilities.}
\wh{A proposal is accepted by the scalar rule of \citet{jiangMaZhang2024} at the constant target $(3+e^{-2})^{-1}$.}
\wh{This benchmark corresponds to the scalar reduction analyzed earlier.}
\item \wh{\textbf{NR.} It follows Algorithm~1 of \citet{naoriRaz2019}: an initial sampling phase, then heavy matching on the observed prefix and randomized rounding of an LP for smaller requests.}
\item \wh{\textbf{Greedy.} On arrival it assigns the request to a currently feasible pool of highest reward, breaking ties by pool index.}
\wh{It measures the reward obtained from immediate allocation based on current feasibility and value.}
\end{itemize}

\wh{VABF-TR and RJMZ use the known arrival laws, whereas NR and Greedy operate without arrival probabilities.}
\wh{For NR, comparisons under the nonstationary laws report empirical performance.}
Appendix~\ref{app:experimental-verification} \wh{gives the implementation details.}

Figure~\ref{fig:bfr-benchmarks-T} \wh{shows how performance changes across six demand levels and three resource configurations, with pool budgets held fixed.}
\begin{figure}[!htbp]
\centering
\includegraphics[width=.32\textwidth]{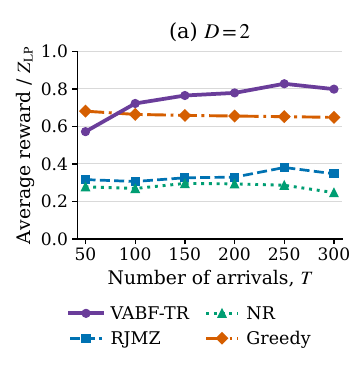}\hfill
\includegraphics[width=.32\textwidth]{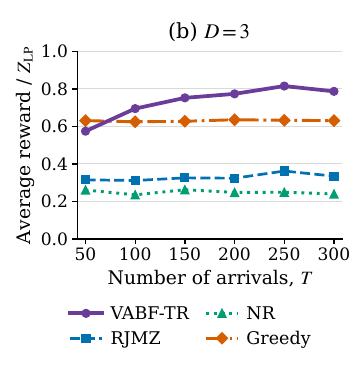}\hfill
\includegraphics[width=.32\textwidth]{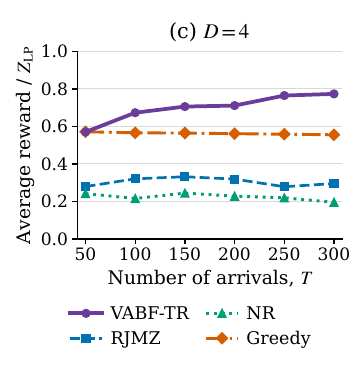}
\caption{\wh{VABF-TR and benchmark performance across demand levels, $K=20$, $S=5$, and $\lambda=3$. Points report average total reward over shared paths, normalized by the corresponding instance's numerically computed $Z_{\mathrm{LP}}$.}}
\label{fig:bfr-benchmarks-T}
\end{figure}

\wh{With budgets held fixed, VABF-TR captures a larger fraction of the LP benchmark at higher demand levels, with fluctuations across adjacent horizons.}
\wh{At $T=100$, its average reward is 72.25\%, 69.51\%, and 67.31\% for $D=2,3,4$, respectively.}
\wh{By $T=300$, these values reach 79.88\%, 78.71\%, and 77.35\%.}
\wh{VABF-TR achieves a higher average reward than RJMZ and NR at every horizon in the figure.}

\wh{As more sessions compete for the same nonreplenishing budgets, an early assignment can exhaust a coordinate needed by later compatible sessions.}
\wh{VABF-TR preserves the assignments of its VABF-T shadow policy and recovers additional requests when they can coexist with those assignments.}
\wh{At the lower demand level $T=50$, Greedy has the higher average reward for $D=2,3$, and the two policies are close for $D=4$.}
\wh{With fewer competing requests, preserving capacity has less value relative to admitting a feasible session immediately.}

\wh{Moving from $D=2$ to $D=3,4$ adds budgets for tool calls and LLM calls to the input- and output-token budgets.}
\wh{For every displayed demand level $T\ge100$, VABF-TR's advantage over Greedy widens across these three resource configurations.}
\wh{At $T=300$, the gaps are 15.03, 15.62, and 21.81 percentage points, respectively.}
\wh{Each added budget can restrict compatible assignments and the capacity available for later sessions.}
\wh{This pattern supports the value of coordinating acceptance across resource budgets while preserving capacity for future demand.}

Figure~\ref{fig:bfr-benchmarks-S} \wh{compares the pooled i.i.d.\ law with the two nonstationary laws at fixed aggregate demand.}
\begin{figure}[!htbp]
\centering
\includegraphics[width=.32\textwidth]{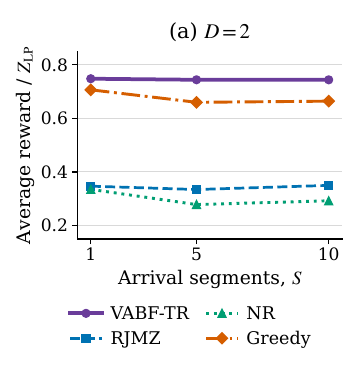}\hfill
\includegraphics[width=.32\textwidth]{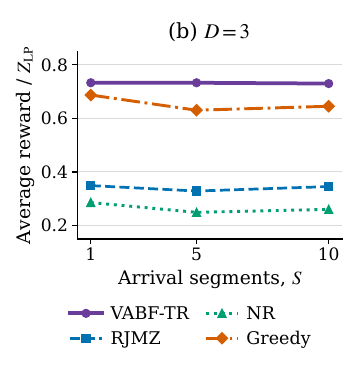}\hfill
\includegraphics[width=.32\textwidth]{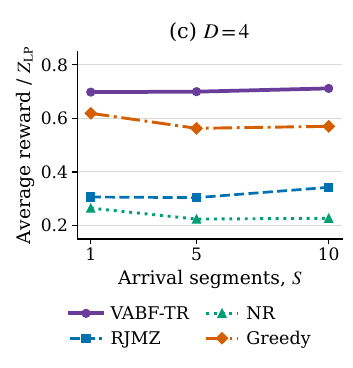}
\caption{\wh{VABF-TR and benchmark performance under pooled i.i.d.\ ($S=1$) and nonstationary ($S=5,10$) arrivals, $K=20$ and $\lambda=3$. \hh{Each point is an equally weighted average of normalized rewards over $T=50,100,150,200,250,300$.} At each demand level, the three laws share expected type counts, budgets, rewards, the LP solution, and $Z_{\mathrm{LP}}$.}}
\label{fig:bfr-benchmarks-S}
\end{figure}
\FloatBarrier

\wh{In Figure~\ref{fig:bfr-benchmarks-S}, VABF-TR's average reward varies by at most 1.38 percentage points across arrival laws at each $D$, while Greedy's decreases by 4.19 to 5.67 points under the nonstationary laws.}
\wh{This pattern highlights the value of preserving capacity for future requests when the request mix varies under fixed budgets.}

\paragraph{Other robustness checks.}
\wh{Additional results for alternative catalogue sizes $K$ and reward weights $\lambda$ are reported in} Appendices~\ref{app:complete-experiment-results} and~\ref{app:benchmark-results}.

\FloatBarrier
\label{sec:computations-end}

\section{Conclusion}
\label{sec:conclusion}
\wh{We study multidimensional online generalized assignment under independent arrivals with known, potentially time-varying type distributions, motivated by admission and allowance commitments in LLM-agent services.}
\wh{VABF combines an LP relaxation with coordinatewise capacity and heavy-assignment count constraints, LP-guided candidate routing, and calibrated acceptance based on the joint distribution of actual vector loads.}
\wh{This design retains resource complementarity in feasibility decisions and links the treatment of large commitments in the LP to the acceptance analysis.}

\wh{Our analysis gives the competitive ratio $R(D)$ relative to the strengthened LP and hence to the expected clairvoyant offline optimum.}
\wh{This} guarantee strictly improves on the scalar JMZ reduction for every $D\ge2$ and achieves the best possible inverse-dimensional order.
\wh{On all-heavy instances, VABF with an appropriate acceptance parameter attains the optimal competitive ratio $1/(D+1)$.}
\hh{We also establish an upper bound of $1/(D+2)$ on the ratio that any online policy can guarantee relative to the strengthened LP.}
\wh{These results identify both the benefit of retaining vector information and the unavoidable dimensional limits of online allocation.}

\wh{Two guarantee-preserving heuristics extend this framework: VABF-T adjusts routing and acceptance parameters, and VABF-TR adds recovery.}
\wh{Numerical results on trace-calibrated instances show favorable performance against benchmarks in most tested settings, especially when capacity is scarce relative to demand.}
\wh{For service providers, the central implication is that admission and assignment decisions should account for the composition of resource commitments and the anticipated evolution of demand.}
\wh{Retaining resource composition permits complementary tasks to share capacity, while calibrated acceptance balances present admissions with future allocation opportunities.}

\FloatBarrier
\phantomsection
\addcontentsline{toc}{section}{References}
\bibliographystyle{plainnat}
\bibliography{refs}

@article{jiangMaZhang2024,
  author  = {Jiang, Jiashuo and Ma, Will and Zhang, Jiawei},
  title   = {Tight Guarantees for Multiunit Prophet Inequalities and Online Stochastic Knapsack},
  journal = {Operations Research},
  year    = {2024},
  volume  = {73},
  number  = {3},
  pages   = {1703--1721},
  doi     = {10.1287/opre.2022.0309},
  url     = {https://arxiv.org/abs/2107.02058v4},
  note    = {Published online June 3, 2024; full version: arXiv:2107.02058v4}
}

@inproceedings{maMacRuryStein2025,
  author    = {Ma, Will and MacRury, Calum and Stein, Cliff},
  title     = {Forward-Backward Contention Resolution Schemes for Fair Rationing},
  booktitle = {Proceedings of the 26th ACM Conference on Economics and Computation},
  year      = {2025},
  pages     = {945},
  publisher = {Association for Computing Machinery},
  doi       = {10.1145/3736252.3742654},
  url       = {https://arxiv.org/abs/2502.09521},
  note      = {Full version: arXiv:2502.09521}
}

@incollection{alaeiEtAl2013,
  author    = {Alaei, Saeed and Hajiaghayi, MohammadTaghi and Liaghat, Vahid},
  title     = {The Online Stochastic Generalized Assignment Problem},
  booktitle = {Approximation, Randomization, and Combinatorial Optimization.
               Algorithms and Techniques},
  series    = {Lecture Notes in Computer Science},
  volume    = {8096},
  year      = {2013},
  pages     = {11--25},
  publisher = {Springer},
  doi       = {10.1007/978-3-642-40328-6_2}
}

@article{albersEtAl2021,
  author  = {Albers, Susanne and Khan, Arindam and Ladewig, Leon},
  title   = {Improved Online Algorithms for Knapsack and {GAP} in the Random
             Order Model},
  journal = {Algorithmica},
  year    = {2021},
  volume  = {83},
  number  = {6},
  pages   = {1750--1785},
  doi     = {10.1007/s00453-021-00801-2}
}

@article{klimmKnaack2026,
  author  = {Klimm, Max and Knaack, Martin},
  title   = {Generalized Assignment and Knapsack Problems in the Random-Order Model},
  journal = {Mathematical Programming},
  year    = {2026},
  doi     = {10.1007/s10107-026-02405-6},
  note    = {Published online August 12, 2026}
}

@inproceedings{naoriRaz2019,
  author    = {Naori, David and Raz, Danny},
  title     = {Online Multidimensional Packing Problems in the Random-Order
               Model},
  booktitle = {30th International Symposium on Algorithms and Computation},
  series    = {Leibniz International Proceedings in Informatics},
  year      = {2019},
  volume    = {149},
  pages     = {10:1--10:15},
  publisher = {Schloss Dagstuhl--Leibniz-Zentrum f{\"u}r Informatik},
  doi       = {10.4230/LIPIcs.ISAAC.2019.10}
}

@incollection{yoshinagaKawase2025,
  author    = {Yoshinaga, Toru and Kawase, Yasushi},
  title     = {Online Contention Resolution Schemes for Size-Stochastic Knapsacks},
  booktitle = {{WALCOM}: Algorithms and Computation},
  series    = {Lecture Notes in Computer Science},
  volume    = {15411},
  year      = {2025},
  pages     = {393--408},
  publisher = {Springer},
  doi       = {10.1007/978-981-96-2845-2_25}
}

@inproceedings{dickersonEtAl2019,
  author    = {Dickerson, John P. and Sankararaman, Karthik Abinav and Sarpatwar, Kanthi Kiran and Srinivasan, Aravind and Wu, Kun-Lung and Xu, Pan},
  title     = {Online Resource Allocation with Matching Constraints},
  booktitle = {Proceedings of the 18th International Conference on Autonomous Agents and Multiagent Systems},
  year      = {2019},
  pages     = {1681--1689},
  publisher = {International Foundation for Autonomous Agents and Multiagent Systems},
  url       = {https://www.ifaamas.org/Proceedings/aamas2019/pdfs/p1681.pdf}
}

@article{wangYanBei2022,
  author = {Wang, Hao and Yan, Zhenzhen and Bei, Xiaohui},
  title = {A Nonasymptotic Analysis for Re-solving Heuristic in Online Matching},
  journal = {Production and Operations Management},
  year = {2022},
  volume = {31},
  number = {8},
  pages = {3096--3124},
  doi = {10.1111/poms.13738}
}

@misc{zhuEtAl2026TraceLab,
  author = {Zhu, Kan and Jacob, Mathew and Ma, Chenxi and Pan, Yi and Wang, Stephanie and Krishnamurthy, Arvind and Kasikci, Baris},
  title = {{TraceLab}: Characterizing Coding Agent Workloads for {LLM} Serving},
  year = {2026},
  howpublished = {arXiv:2606.30560, version 2},
  doi = {10.48550/arXiv.2606.30560},
  url = {https://arxiv.org/abs/2606.30560v2}
}

@article{agrawalEtAl2014,
  author  = {Agrawal, Shipra and Wang, Zizhuo and Ye, Yinyu},
  title   = {A Dynamic Near-Optimal Algorithm for Online Linear Programming},
  journal = {Operations Research},
  year    = {2014},
  volume  = {62},
  number  = {4},
  pages   = {876--890},
  doi     = {10.1287/opre.2014.1289}
}

@article{kesselheimEtAl2018,
  author  = {Kesselheim, Thomas and Radke, Klaus and T{\"o}nnis, Andreas and V{\"o}cking, Berthold},
  title   = {Primal Beats Dual on Online Packing {LPs} in the Random-Order Model},
  journal = {SIAM Journal on Computing},
  year    = {2018},
  volume  = {47},
  number  = {5},
  pages   = {1939--1964},
  doi     = {10.1137/15M1033708}
}

@incollection{feldmanEtAl2010,
  author    = {Feldman, Jon and Henzinger, Monika and Korula, Nitish and Mirrokni, Vahab S. and Stein, Cliff},
  title     = {Online Stochastic Packing Applied to Display Ad Allocation},
  booktitle = {Algorithms--ESA 2010},
  series    = {Lecture Notes in Computer Science},
  volume    = {6346},
  year      = {2010},
  pages     = {182--194},
  publisher = {Springer},
  doi       = {10.1007/978-3-642-15775-2_16}
}

@article{molinaroRavi2014,
  author  = {Molinaro, Marco and Ravi, R.},
  title   = {The Geometry of Online Packing Linear Programs},
  journal = {Mathematics of Operations Research},
  year    = {2014},
  volume  = {39},
  number  = {1},
  pages   = {46--59},
  doi     = {10.1287/moor.2013.0612}
}

@article{liYe2022,
  author  = {Li, Xiaocheng and Ye, Yinyu},
  title   = {Online Linear Programming: Dual Convergence, New Algorithms, and Regret Bounds},
  journal = {Operations Research},
  year    = {2022},
  volume  = {70},
  number  = {5},
  pages   = {2948--2966},
  doi     = {10.1287/opre.2021.2164}
}

@article{feldmanSvenssonZenklusen2016,
  author    = {Feldman, Moran and Svensson, Ola and Zenklusen, Rico},
  title     = {Online Contention Resolution Schemes with Applications to {Bayesian} Selection Problems},
  year      = {2021},
  pages     = {255--300},
  doi       = {10.1137/18M1226130},
  journal   = {SIAM Journal on Computing},
  volume    = {50},
  number    = {2}
}

@misc{aoEtAl2025Inference,
  author = {Ao, Ruicheng and Luo, Gan and Simchi-Levi, David and Wang, Xinshang},
  title = {Optimizing {LLM} Inference: Fluid-Guided Online Scheduling with Memory Constraints},
  year = {2025},
  howpublished = {arXiv:2504.11320, version 4},
  note = {Revised June 2026},
  doi = {10.48550/arXiv.2504.11320},
  url = {https://arxiv.org/abs/2504.11320v4}
}

@inproceedings{nieSiZhou2026Inference,
  author    = {Nie, Chengyi and Si, Nian and Zhou, Zijie},
  title     = {A Queueing-Theoretic Framework for Stability Analysis of {LLM} Inference with {KV} Cache Memory Constraints},
  booktitle = {Proceedings of the 43rd International Conference on Machine Learning},
  series    = {Proceedings of Machine Learning Research},
  volume    = {306},
  pages     = {93016--93031},
  publisher = {PMLR},
  year      = {2026},
  url       = {https://proceedings.mlr.press/v306/nie26a.html}
}

@misc{chenYeZhou2026Routing,
  author = {Chen, Zixi and Ye, Yinyu and Zhou, Zijie},
  title = {Online Linear Programming for Multi-Objective Routing in {LLM} Serving},
  year = {2026},
  howpublished = {arXiv:2607.03948},
  doi = {10.48550/arXiv.2607.03948},
  url = {https://arxiv.org/abs/2607.03948}
}

\StartAppendix
Appendix~\ref{app:core-certificates} gives the dimension-preprocessing argument and proves Lemmas~\ref{lem:oracle-baseline}--\wh{\ref{lem:gamma-monotonicity}}, \wh{Proposition~\ref{lem:bf-targets},} and Theorem~\ref{thm:discrete-main}.
Appendix~\ref{app:implementation-rounding} details the state implementation and proves the upward-rounding guarantees.
Appendices~\ref{app:hardness-details} and~\ref{app:lp-hardness} prove Theorems~\ref{thm:limit-heavy} and~\ref{thm:lp-online-limit}, respectively.
Appendix~\ref{app:computational-details} supplies details and complete results for the \wh{numerical} study in Section~\ref{sec:computations}.

\section{{Proofs for the Model and Vector-Aware Best-Fit}}
\label{app:start}
\label{app:core-certificates}
\subsection{Dimension preprocessing}
\label{sec:dimension-preprocessing}
For each server, retain one representative of each identical coordinate profile and delete any profile dominated by another distinct profile.
Every deleted nonzero profile is dominated by a retained profile: following strict dominance in this finite collection reaches a maximal profile.
{Thus the preprocessing preserves every feasible allocation,} because cumulative consumption in a deleted coordinate is bounded by that in a retained coordinate.
Pairwise comparison takes $O(|E_u|m_u^2)$ operations, where $m_u$ is the number of coordinates before preprocessing.
The maximum scalar size of every edge is unchanged.
In particular, a common maximum coordinate reduces the server's budget vector to one resource dimension.

\subsection{{Proof of Lemma~\ref{lem:oracle-baseline}}}
\label{sec:oracle-proof}
\proof{Proof of Lemma~\ref{lem:oracle-baseline}.}
{Index the retained coordinates at server {$u$} by {$1,\ldots,|\mathcal D_u|$}, where {$|\mathcal D_u|\le D$}.
For every edge {$(u,v)$}, choose a maximizing index {$d(u,v)\in\arg\max_d s_{uv}^d$} with fixed tie-breaking.
Fix a realized request sequence {$\omega=(v_1,\ldots,v_T)$} and an optimal vector-feasible offline allocation {$S^*(\omega)$}.
Represent an assignment by {$(t,u)$}, so repeated arrivals of the same type remain distinct.
Partition the allocation into {\wh{$S_j(\omega)=\{(t,u)\in S^*(\omega):d(u,v_t)=j\}, \quad j=1,\ldots,D$.}}
Each group assigns every request at most once.
At a server with {$j\le|\mathcal D_u|$}, its scalar consumption satisfies
{\[
 \sum_{t:(t,u)\in S_j(\omega)}b_{uv_t}
 =\sum_{t:(t,u)\in S_j(\omega)}s_{uv_t}^{j}
 \le\sum_{t:(t,u)\in S^*(\omega)}s_{uv_t}^{j}\le1.
\]}
If {$j>|\mathcal D_u|$}, the group contains no assignment to that server.
Hence every group is feasible in {$I_{\mathrm{sc}}$}.
Their rewards sum to the optimal vector reward, so the most valuable group earns at least {$1/D$} of that reward.
Taking expectations yields {$\OPT_{\mathrm{sc}}\ge\OPT(I)/D$}.
This partition is used only to compare offline values; the online policy does not select a group using future arrivals.

The scalar surrogate preserves the arrival process and information structure of the original instance.
Running {$\mathcal A$} on it therefore yields a nonanticipatory policy.
Implementing its assignments in the vector instance preserves reward and feasibility, since each coordinate consumption is bounded by the corresponding scalar consumption.
By the assumed guarantee, {$\E[\ALG_{\mathrm{red}}(I)]\ge\alpha\OPT_{\mathrm{sc}}$}.
Combining the two comparisons \wh{gives $\E[\ALG_{\mathrm{red}}(I)]\ge(\alpha/D)\OPT(I)$}, including {$D=1$}, {$\alpha=0$}, and {$\OPT(I)=0$}.}\Halmos\endproof



{\subsection{Proof of Lemma~\ref{lem:planning-upper}}
\label{sec:lp-proof}
\proof{Proof.}
For a realized offline optimum, let {$N_{uv}^*$} be the number of assignments on edge {$(u,v)$} and {$N_v$} the number of type-{$v$} arrivals.
The assignment constraint gives {$\sum_uN_{uv}^*\le N_v$}.
For every coordinate, feasibility and the strict heavy boundary give {\wh{$\sum_{v:(u,v)\in E}s_{uv}^dN_{uv}^*\le1,\quad \sum_{v\in V_{ud}^{H}}N_{uv}^*\le1$.}}
Taking expectations and using {$\E N_v=\mu_v$} shows that \wh{the vector $\mathbf z$ defined by} {$z_{uv}=\E[N_{uv}^*]$} is LP feasible and has the expected offline reward.}\Halmos\endproof

\subsection{{Proof of Lemma~\ref{lem:bf-marginal}}}
\label{sec:marginal}
\proof{Proof.}
{Initially {$\mathbf L_0=\mathbf0$}.
The vector distribution has finite support because the horizon and the outcome set at each epoch are finite.
At each epoch {$i\le t$}, a positive target and \eqref{eq:bf-fitting-condition} ensure that the feasible maximum-load levels have a positive-mass boundary satisfying \eqref{eq:bf-threshold}.
The coin in \eqref{eq:bf-boundary} lies in {$[0,1]$} and makes the selected unconditional mass exactly {$\gamma_i$}.
A zero target rejects all proposals without a boundary.
Acceptance is restricted to {$\boldsymbol{\ell}+\mathbf s_e\le\mathbf1$}; rejection and null outcomes leave the vector load unchanged.
Hence every accepted assignment respects all coordinate capacities on every realization.

The current categorical outcome is independent of the pre-epoch vector state.
Conditional on edge {$e$} being proposed at epoch {$i$}, with {$q_{e,i}>0$}, the selected states and boundary coin therefore give acceptance probability {$\gamma_i$} and expected reward {$\gamma_iw_e$}.
Let {$Y_i$} be this server's reward at epoch {$i$}, equal to zero if no proposal is accepted.
Averaging over outcomes yields {$\E Y_i=\sum_e q_{e,i}\gamma_iw_e=\gamma_i Z_i$}.
Each accepted edge adds {$b_e$} to {$X$}, its heavy indicator to {$J$}, and {$s_e^d$} to coordinate {$d$}.
Consequently, {\wh{$\E L_t^d=\sum_{i\le t,e}\gamma_iq_{e,i}s_e^d$,}} and the same summation for scalar sizes, heavy indicators, and rewards gives \eqref{eq:bf-marginals}.
These identities do not require {$X_t\le1$} or {$J_t\le1$}.
Constant targets give expected reward {$\gamma Z$}.}
\Halmos\endproof

\subsection{Proof of Lemma~\ref{lem:bf-jmz}}
\label{sec:load-refinement}
\label{app:jmz-induction}
\proof{Proof.}
Use the single-server notation from Section~\ref{sec:single-server} and the assumptions of the lemma, so every target satisfies $\gamma_i\le c$ and $c>0$.
\wh{We prove \eqref{eq:bf-H-definition} first, then derive \eqref{eq:bf-refined-empty}.}

\emph{Maximum-load invariant.}
For the first bound, we extend the scalar utilization argument of \citet[Lemma~4 and Theorem~5 in the full version]{jiangMaZhang2024} to vector loads.
If edge $e$ is accepted at epoch $i$, then \wh{$\max\{M_{i-1},b_e\}\le M_i\le M_{i-1}+b_e$.}
Thus maximum load never decreases, and every state with $M_{i-1}+b_e\le1$ is vector feasible.
Define $Q_i(x)=\Prob(0<M_i\le x)$.
We prove by induction that
\begin{equation}
 Q_i(x)\le c\exp\!\left[-\frac{Q_i(1-x)-Q_i(x)}c\right],
 \qquad 0<x\le\tfrac12,\quad 0\le i\le t.
 \label{eq:jmz-categorical-invariant}
\end{equation}
Fix $x$ and an epoch $i$, write \wh{$F=Q_{i-1}(x),\quad G=Q_{i-1}(1-x)-Q_{i-1}(x)$,} and use $F',G'$ for the corresponding unconditional masses after the epoch.
Set $\gamma=\gamma_i$.
Conditional on a positive-probability proposal, independence preserves the pre-epoch state distribution.
Write $b=b_e$ and say that its selection \emph{uses zero} when it accepts positive mass from the empty state; all positive fitting states are then selected first.
Partition the nonnull proposals into four groups:
\begin{enumerate}
\item $b\le x$ and the selection uses zero;
\item $x<b\le1-x$ and the selection uses zero;
\item $b\le1-x$ and the selection does not use zero;
\item $b>1-x$.
\end{enumerate}
Let $p_1,p_2$ be the probabilities of the first two groups, and let $A_1,A_3$ be the probability-weighted masses leaving $(0,x]$ in groups 1 and 3, with $A=A_1+A_3$.
In group 1, all states in $(0,1-x]$ fit, so at most $\gamma-F-G$ mass enters $(0,x]$ from zero.
In group 2, all states in $(0,x]$ fit and are selected, while at most $\gamma$ mass enters $(x,1-x]$.
In groups 1 and 3, new mass in $(x,1-x]$ can only come from departures out of $(0,x]$.
Group 3 supplies no new mass from zero, group 4 adds no mass to either interval, and a null outcome changes nothing.
Averaging over proposals gives
\begin{align}
 F'&\le F+(\gamma-F-G)p_1-Fp_2-A,\label{eq:jmz-L-transfer}\\
 G'&\le G+\gamma p_2+A,\label{eq:jmz-R-transfer}
\end{align}
where $p_1+p_2\le1$.

The invariant holds initially because $Q_0=0$.
Suppose it holds before epoch $i$.
If $p_1>0$, a group-1 proposal uses zero, implying $F+G\le\gamma\le c$.
Then \eqref{eq:jmz-L-transfer} and $p_1\le1-p_2$ yield \wh{$F'\le(\gamma-G)(1-p_2)-A\le(c-G)(1-p_2)-A$.}
Using $e^{-(r+\delta)}\ge e^{-r}-\delta$ for $r,\delta\ge0$, \eqref{eq:jmz-R-transfer}, and $0\le G/c,p_2\le1$, we obtain
\[
 e^{-G'/c}\ge e^{-G/c-p_2}-\frac Ac
 \ge(1-G/c)(1-p_2)-\frac Ac
 \ge\frac{F'}c.
\]
If $p_1=0$, then $A_1=0$ and \wh{$F'\le F(1-p_2)-A_3,\quad G'\le G+cp_2+A_3$.}
The induction hypothesis gives
\[
 e^{-G'/c}\ge e^{-G/c}(1-p_2)-\frac{A_3}c
 \ge\frac Fc(1-p_2)-\frac{A_3}c
 \ge\frac{F'}c.
\]
This proves \eqref{eq:jmz-categorical-invariant}, including boundary randomization through the selected masses.

\emph{The first bound.}
Let $p_0=\Prob(M_t=0)$.
Since $M_t\in[0,1]$,
\begin{equation}
 \E M_t=1-p_0-\int_0^1Q_t(s)\,ds.
 \label{eq:jmz-paired-integral-identity}
\end{equation}
If $1-p_0\le c$, then $p_0\ge1-c-\E M_t$, proving \eqref{eq:bf-H-definition}.
Otherwise, choose the unique $\eta\in(0,1)$ satisfying $1-p_0=c(\eta-\log\eta)$; uniqueness follows because $\eta-\log\eta$ decreases from infinity to one.
For $0<s\le1/2$, we claim
\begin{equation}
 Q_t(s)+Q_t(1-s)\le\max\{1-p_0+c\eta,\,2c\}.
 \label{eq:jmz-pair-bound}
\end{equation}
If $Q_t(s)<c\eta$, use $Q_t(1-s)\le1-p_0$.
Otherwise \eqref{eq:jmz-categorical-invariant} gives
\[
 Q_t(s)+Q_t(1-s)\le2Q_t(s)-c\log(Q_t(s)/c),
 \qquad c\eta\le Q_t(s)\le Q_t(1/2)\le c.
\]
The right-hand side is convex in $Q_t(s)$, so its maximum on $[c\eta,c]$ occurs at an endpoint, giving \eqref{eq:jmz-pair-bound}.
Pairing the two halves of the integral in \eqref{eq:jmz-paired-integral-identity} therefore yields \wh{$\E M_t\ge1-p_0-\max\{(1-p_0+c\eta)/2,\,c\}$.}
If the maximum is $c$, then $1-p_0\le c+\E M_t$.
Otherwise $-c\log\eta=1-p_0-c\eta\le2\E M_t$, so $\eta\ge e^{-2\E M_t/c}$ and monotonicity of $\eta-\log\eta$ gives \wh{$1-p_0=c(\eta-\log\eta)\le2\E M_t+ce^{-2\E M_t/c}$.}
Either case proves \eqref{eq:bf-H-definition}.

\emph{Interval-mass bounds.}
For \wh{the second bound \eqref{eq:bf-refined-empty}}, define, for $0\le i\le t$,
\begin{equation}
 \begin{split}
 \lambda_{1,i}&=\Prob(0<M_i\le\tfrac13),\qquad
 \lambda_{2,i}=\Prob(\tfrac13<M_i\le\tfrac12),\\
 \lambda_{3,i}&=\Prob(J_i=0,\ \tfrac12<M_i\le X_i\le\tfrac23).
 \end{split}
 \label{eq:bf-interval-masses}
\end{equation}
A heavy acceptance makes $M_i>1/2$ thereafter, so $J_i=0$ on the first two events.
The third event also restricts $X_i$, which controls its transitions.
We prove
\begin{equation}
 \lambda_{1,i}+\lambda_{2,i}\le c,\qquad
 \lambda_{1,i}+\lambda_{3,i}\le2c
 \label{eq:bf-interval-invariants}
\end{equation}
by induction from the empty state.
For one step, abbreviate $\lambda_j=\lambda_{j,i-1}$ and condition on a positive-probability proposal $e$, with $b=b_e$ and target $\gamma=\gamma_i\le c$.
Primes denote post-step masses conditional on this proposal; the pre-step distribution is unchanged by independence.

For the first inequality, if $b>1/2$, acceptance leaves both intervals, so their total mass cannot increase.
If $b\le1/2$, every state in their union fits and precedes zero, the only source of new mass.
Thus
\[
 \lambda_1'+\lambda_2'
 \le\lambda_1+\lambda_2+(\gamma-\lambda_1-\lambda_2)^+
 =\max\{\lambda_1+\lambda_2,\gamma\}\le c.
\]
For the second inequality, consider three cases.
If $b\le1/3$, all states counted by $\lambda_3$ fit because $M_{i-1}\le X_{i-1}\le2/3$.
New mass in the union defining $\lambda_1+\lambda_3$ can only come from zero or $(1/3,1/2]$.
Indeed, states with $J_{i-1}>0$, with $M_{i-1}>2/3$, or with $M_{i-1}>1/2$ and $X_{i-1}>2/3$ cannot enter that union, since loads and the heavy-assignment count never decrease.
All possible outside sources follow the $\lambda_3$ states in maximum-load order.
Their selected mass is therefore at most $(\gamma-\lambda_3)^+$, giving
\[
 \lambda_1'+\lambda_3'
 \le\lambda_1+\lambda_3+(\gamma-\lambda_3)^+
 =\max\{\lambda_1+\lambda_3,\lambda_1+\gamma\}\le2c.
\]
Here $\lambda_1\le c$ follows from the first interval bound.
If $1/3<b\le1/2$, no accepted state enters $(0,1/3]$.
A new third-event state must have $X_{i-1}\le2/3-b<1/3$ before acceptance and cannot originate at zero, since acceptance from zero gives maximum load $b\le1/2$.
It therefore comes from the first event, so the union's mass cannot increase.
If $b>1/2$, acceptance gives maximum load above $1/2$ and a positive heavy-assignment count, creating neither event.
Null outcomes change nothing; averaging over proposals proves \eqref{eq:bf-interval-invariants}.

\emph{Pointwise weights and complementarity.}
For the remainder, write $\lambda_j=\lambda_{j,t}$.
By \eqref{eq:bf-complementarity}, \wh{$X_t-\Psi_t=\max\{M_t,\min(X_t,2/3)\}$.}
On a nonempty path with $J_t\ge1$, we have $M_t>1/2$, so the adjusted weight $\tfrac32(X_t-\Psi_t)+\tfrac14J_t$ is at least one.
On a nonempty path with $J_t=0$ outside the three events in \eqref{eq:bf-interval-masses}, we have $M_t>1/2$ and $X_t>2/3$, so the weight is again at least one.
Inside the first two events, $X_t-\Psi_t\ge M_t$ bounds the residual below unit weight by $1$ and $1/2$, respectively.
Inside the third, $\Psi_t=0$ and $X_t>1/2$ bound the residual by $1/4$.
The empty path has zero weight, hence
\[
 1-p_0\le\tfrac32\E X_t+\tfrac14\E J_t-\tfrac32\E\Psi_t
              +\lambda_1+\tfrac12\lambda_2+\tfrac14\lambda_3.
\]

\emph{Bounding the residual.}
The event inclusion \wh{$Q_t(2/3)-Q_t(1/3)\ge\lambda_2+\lambda_3$} and \eqref{eq:jmz-categorical-invariant} at $x=1/3$ give $\lambda_1\le c e^{-(\lambda_2+\lambda_3)/c}$.
If $\lambda_1=0$, \eqref{eq:bf-interval-invariants} bounds the residual by $c$.
Otherwise set $r=\lambda_1/c\in(0,1]$.
The interval bounds and the exponential inequality give
\[
 \lambda_2\le c(1-r),\qquad
 \lambda_3\le c(2-r),\qquad
 \lambda_2+\lambda_3\le-c\log r.
\]
Combining the first inequality with each of the other two yields
\[
 \frac{\lambda_1+\lambda_2/2+\lambda_3/4}{c}
 \le\min\left\{1+\frac r4,\ \frac{1+3r-\log r}{4}\right\}.
\]
Let $\rho\in(0,1)$ be the unique root of $\rho=e^{-3+2\rho}$.
Uniqueness follows because $\log r-2r+3$ increases from $-\infty$ to $2-\log2>0$ on $(0,1/2]$, then decreases to $1$ on $[1/2,1]$.
For $r\le\rho$, the first branch is at most $1+\rho/4$.
For $r\in[\rho,1]$, the second branch is convex with endpoint values $1+\rho/4$ and $1$, so the same bound holds.
Consequently,
\[
 p_0\ge1-\tfrac32\E X_t-\tfrac14\E J_t-(1+\rho/4)c+\tfrac32\E\Psi_t.
\]
Numerical evaluation gives $\rho\approx0.0556483<0.056$.
Thus $1+\rho/4<1.014$, proving \eqref{eq:bf-refined-empty}. 
\Halmos\endproof

\subsection{Proof of Lemma~\ref{lem:pair-credit}}
\label{sec:pair-credit-proof}
\proof{Proof.}
\wh{If $\varepsilon=0$, the claim follows from $\Psi_t\ge0$.
Suppose $\varepsilon>0$.
The first two epochs are feasible by assumption.
After epoch one, the nonempty mass is $\gamma\sum_e q_{e,1}\le\gamma$.
Independence of proposals preserves this pre-epoch distribution conditional on any second-epoch proposal.}
All feasible positive states at epoch two are fully selected because their total mass is at most {$\gamma$} and they precede zero.
Choose a pair {$e,f$} attaining \eqref{eq:bf-pair-credit}.
Since {$b_e+b_f\le1$}, the second edge fits after the first.
Conditional on the first two proposals being {$e,f$}, the first is accepted with probability {$\gamma$}, and the second is accepted whenever the first was accepted.
The event that the first two proposals are {$e,f$} and both are accepted therefore has probability {$\gamma q_{e,1}q_{f,2}$}.
\wh{On this event, $X_2=b_e+b_f$ and $M_2=\|\mathbf s_e+\mathbf s_f\|_\infty$, so $\Psi_2$ equals the minimum in \eqref{eq:bf-pair-credit}.
All other paths have nonnegative compensation, giving} {\wh{$\E\Psi_2\ge\gamma\varepsilon$.}}
Each acceptance satisfies {$0\le M_t-M_{t-1}\le b_e$}.
Hence {$X_t-M_t$} and {$(X_t-2/3)^+$} are both nondecreasing on every path, and so is {$\Psi_t$}.
Monotonicity gives this lower bound at every subsequent time {$t\ge2$} for which epochs {$1,\ldots,t$} are feasible.
\Halmos\endproof

\subsection{{Proof of Proposition~\ref{lem:bf-targets}}}
\label{sec:target-proof}
\proof{Proof.}
{\wh{Write} {$B_t=\sum_{i\le t}\psi_i$}, with {$B_0=0$} and {$B_T=B$}.
Nonnegative consumption gives {\wh{$0\le b_e+b_f-\|\mathbf s_e+\mathbf s_f\|_\infty\le b_f$.}}
Thus every candidate in \eqref{eq:bf-pair-credit} is nonnegative and at most {$q_{f,2}b_f$}.
When {$T\ge2$}, {\wh{$0\le\varepsilon\le\psi_2\le B-B_t\quad(t=0,1)$.}}
For {$T<2$} or no eligible pair, {$\varepsilon=0$}.
In all cases {$0\le\varepsilon\le B$}.

The zero target rejects every proposal and earns zero reward.
Fix {$0<\gamma\le\Gamma(B-\varepsilon,\theta)$} and put {$\bar B=B-\varepsilon$}.

\emph{Compensation through feasible epochs.}
Use {$\Psi_t$} from \eqref{eq:bf-complementarity} and define, only for analysis, {\wh{$\widehat X_t=X_t-\Psi_t =\max\{M_t,\min(X_t,2/3)\}$.}}
\wh{For $t\ge2$, Lemma~\ref{lem:pair-credit} gives $\E\Psi_t\ge\gamma\varepsilon$ whenever epochs $1,\ldots,t$ are feasible.}
For {$t=0,1$} the remaining-load bound above gives {$B_t\le B-\varepsilon$}.
Together with Lemma~\ref{lem:bf-marginal}, these facts imply, whenever epochs {$1,\ldots,t$} are feasible,
{\begin{equation}
 \E\widehat X_t\le\gamma\bar B,\qquad
 \E J_t\le\gamma\theta.
 \label{eq:bf-effective-load}
\end{equation}}
No part of this argument assumes the feasibility of a later epoch.

\emph{Target induction.}
The first epoch is feasible because {$\Prob(M_0=0)=1\ge\gamma$}.
For {$t\ge2$}, suppose epochs {$1,\ldots,t-1$} are feasible.
Lemma~\ref{lem:bf-jmz} therefore applies at time {$t-1$} with {$c=\gamma$}.
If {$\gamma\le\wh{\Gamma_2}(\bar B,\theta)$}, then \eqref{eq:bf-refined-empty} and \eqref{eq:bf-effective-load} give
{\[
 \Prob(M_{t-1}=0)
 \ge1-\tfrac32\E\widehat X_{t-1}
       -\tfrac14\E J_{t-1}-1.014\gamma
 \ge1-\gamma(3\bar B/2+\theta/4+1.014)\ge\gamma.
\]}
The last inequality is exactly {$\gamma(3\bar B/2+\theta/4+2.014)\le1$}.

If {$\gamma\le\wh{\Gamma_1}(\bar B)$}, use {$M_{t-1}\le\widehat X_{t-1}$} and the fact that {$H_\gamma$} is nonincreasing: its two branches have derivatives {$-1$} and {$-2+2e^{-2y/\gamma}\le0$} for {$y\ge0$}.
Consequently,
{\begin{align*}
 \Prob(M_{t-1}=0)
 &\ge H_\gamma(\E M_{t-1})
 \ge H_\gamma(\gamma\bar B)\\
 &=\min\{1-\gamma-\gamma\bar B,\,
          1-2\gamma\bar B-\gamma e^{-2\bar B}\}
 \ge\gamma.
\end{align*}}
Here the last inequality is equivalent to {$\gamma(2+\bar B)\le1$} and {$\gamma(1+2\bar B+e^{-2\bar B})\le1$}.
At least one sufficient condition for feasibility holds for the chosen target.
The empty state fits every legal vector proposal, so epoch {$t$} is feasible.
Induction proves feasibility throughout the horizon, and Lemma~\ref{lem:bf-marginal} gives expected reward {$\gamma Z$}.

\wh{By Lemma~\ref{lem:gamma-monotonicity} and $B-\varepsilon\le B$, we have $\Gamma(B-\varepsilon,\theta)\ge\Gamma(B,\theta)$, so the original target range is contained in the range guaranteed to be feasible.}}
\Halmos\endproof

\subsection{Proof of Lemma~\ref{lem:network-reward}}
\label{sec:network-reward-proof}
\proof{Proof.}
\wh{Fix the optimal LP solution $\mathbf z$ used for routing.
For each server $u$, the proposals form an independent categorical stream with $\sum_t q_{e,t}=z_e$ for every incident edge $e$.
Proposition~\ref{lem:bf-targets} makes every constant target $0\le\gamma_u\le\Gamma(B_u-\varepsilon_u,\theta_u)$ feasible throughout the horizon and gives expected reward} \wh{$\gamma_u \sum_{t=1}^T\sum_{e\in\delta(u)}q_{e,t}w_e =\gamma_u Z_u$.}
\wh{Every accepted proposal respects its server's capacities, and routing proposes each request to at most one server, so the network allocation is feasible.
Linearity of expectation gives $\E[\ALG_{\mathrm{VABF}}]=\sum_u\gamma_u Z_u$, without requiring independence across servers.
Setting $\gamma_u=\Gamma(B_u-\varepsilon_u,\theta_u)$ gives the equality in \eqref{eq:bf-cr}.
Since $Z_u\ge0$ and $\sum_u Z_u=Z_{\mathrm{LP}}$, this sum is at least $\bigl[\min_u\Gamma(B_u-\varepsilon_u,\theta_u)\bigr]Z_{\mathrm{LP}}$.
The minimum is nonnegative, so Lemma~\ref{lem:planning-upper} gives the final inequality in \eqref{eq:bf-cr}.}
\Halmos\endproof

{\subsection{Proof of Lemma~\ref{lem:local-load}}
\label{sec:local-load-proof}
\proof{Proof.}
Since {$b_{uv}=\max_d s_{uv}^d\le\sum_d s_{uv}^d$}, the capacity rows give
{\[
 B_u\le\sum_{d\in\mathcal D_u}\sum_{v:(u,v)\in E}s_{uv}^dz_{uv}
 \le|\mathcal D_u|\le D.
\]}
Every heavy edge is heavy in at least one coordinate, so the heavy-assignment count rows give
{\[
 \theta_u\le\sum_{d\in\mathcal D_u}\sum_{v\in V_{ud}^{H}}z_{uv}
 \le|\mathcal D_u|\le D.
\]}
The first inequality may count an edge more than once, which only increases its nonnegative contribution to the upper bound.}\Halmos\endproof

\subsection{Proof of Lemma~\ref{lem:gamma-monotonicity}}
\label{sec:gamma-monotonicity-proof}
\proof{Proof.}
\wh{The functions $1+x$ and $2x+e^{-2x}$ are strictly increasing on $[0,\infty)$: their derivatives are $1$ and $2-2e^{-2x}$, respectively, and the latter is positive for $x>0$.
Their maximum is strictly increasing, so the positive denominator of $\Gamma_1(x)$ is strictly increasing and $\Gamma_1$ is strictly decreasing in $x$.
The positive denominator of $\Gamma_2(x,y)=4/(6x+y+8.056)$ is strictly increasing in either argument, so $\Gamma_2$ is strictly decreasing in each argument.
Taking the maximum of these two branches preserves strict decrease in $x$.
Since $\Gamma_1$ is independent of $y$, their maximum is nonincreasing in $y$.}
\Halmos\endproof

\subsection{{Proof of Theorem~\ref{thm:discrete-main}}}
\label{sec:network-proof}
\proof{Proof.}
{Fix the common target {$\gamma_{u,t}=R(D)$}.
By \hyperref[lem:local-load]{Lemma~\ref*{lem:local-load}}, {$B_u,\theta_u\le D$}, so \wh{Lemma~\ref{lem:gamma-monotonicity} gives} {$\Gamma(B_u,\theta_u)\ge\Gamma(D,D)$}.
For integer {$D\ge1$}, {$2D+e^{-2D}\ge1+D$}, giving {$\Gamma(D,D)=R(D)$}.
\wh{Since $0\le\varepsilon_u\le B_u$, Lemma~\ref{lem:gamma-monotonicity} also gives $\Gamma(B_u-\varepsilon_u,\theta_u)\ge\Gamma(B_u,\theta_u)\ge R(D)$.
Lemma~\ref{lem:network-reward} therefore makes the common target $\gamma_{u,t}=R(D)$ feasible throughout the network and gives $\E[\ALG_{\mathrm{VABF}}]=R(D)\sum_u Z_u=R(D)Z_{\mathrm{LP}}$.}
Applying \hyperref[lem:planning-upper]{Lemma~\ref*{lem:planning-upper}} proves \eqref{eq:adaptive-reward}.}

\Halmos\endproof

\section{Implementation and Upward Rounding}
\label{app:implementation-rounding}
This section supplies the implementation argument for \wh{VABF} in Section~\ref{sec:routing-admission} and the rounding guarantees stated below.
We extend the scalar upward-rounding method of \citet[Section~5.1 of the cited full version]{jiangMaZhang2024} to vector loads, retaining the original LP solution and routing.
The analysis uses exact arithmetic and prescribed exact targets; \wh{the preparation of the experimental policies is described separately in} Appendix~\ref{app:experimental-policies}.

\paragraph{Exact state recursion.}
Initially, $\nu_{u,0}(\mathbf0)=1$, and all other states have zero mass.
Finite types and a finite horizon imply finite reachable support.
At epoch $t$, independence of the current proposal from earlier arrivals and coins makes $\nu_{u,t-1}$ the pre-arrival state law conditional on any positive-probability proposal to $u$.
For each such proposal, group its feasible states by their pre-arrival maximum coordinate load and choose the threshold and common boundary probability in \eqref{eq:bf-threshold}--\eqref{eq:bf-boundary}.
A zero target rejects every proposal.
All thresholds at that epoch use the same distribution $\nu_{u,t-1}$.
A null outcome or a rejected proposal leaves the state unchanged; an accepted proposal adds its resource vector.
Averaging these mutually exclusive outcomes gives exactly \eqref{eq:bf-vector-dp}, with unreachable predecessor states assigned zero mass.
This construction therefore implements the policy whenever the prescribed targets are feasible.

\paragraph{Rounded virtual loads and complexity.}
Fix an integer $m\ge1$ and let $h=1/(2mT)$.
Retain $\wh{\mathbf z}$, all proposal probabilities, and all rewards, and replace each coordinate by \wh{$\widetilde s_e^d=h\left\lceil s_e^d/h\right\rceil, \quad \widetilde b_e=\max_d\widetilde s_e^d$.}
\wh{Write $\widetilde{\mathbf s}_e=(\widetilde s_e^d)_d$ for the rounded consumption vector.}
Run the state recursion and online acceptance rule using these rounded vectors and their virtual loads.
Because $1$ is a grid point, $0\le s_e^d\le\widetilde s_e^d\le1$.
Every virtually feasible allocation is therefore feasible for the original capacities.
Because $1/2$ is also a grid point, $s_e^d>1/2$ if and only if $\widetilde s_e^d>1/2$; equality remains on the nonheavy side.
Thus both coordinatewise and edgewise heavy classifications are unchanged.
Moreover,
\[
 b_e\le\widetilde b_e\le b_e+h,
 \qquad
 \widetilde B_u:=\sum_{t,e\in\delta(u)}q_{e,t}\widetilde b_e
 \le B_u+Th\le D+\frac1{2m},
 \qquad \widetilde\theta_u=\theta_u\le D.
\]
Here each server has total proposal probability at most one in each epoch, and Lemma~\ref{lem:local-load} applies to the original LP solution.
On any realized path, the extra consumption of a coordinate is at most $Th$.
The original LP solution need not be feasible for an LP with rounded coefficients: the following argument uses its unchanged proposal stream and the displayed load bounds.

Each server has at most $N=(2mT+1)^D$ virtual states.
For each epoch and incident edge, checking vector feasibility, accumulating mass by maximum-load level, and updating the distribution requires $O(DN)$ arithmetic operations and comparisons.
The $2mT+1$ grid levels can be scanned in decreasing order, so no sorting of the states is required at each step.
After the LP and routing have been prepared, distributions and thresholds for prescribed targets therefore require \wh{$O\bigl(DT|E|(2mT+1)^D\bigr)$} arithmetic operations and comparisons.
Servers without incident edges can be omitted.
This bound is polynomial in $m,T,|E|$ for fixed $D$, but exponential in $D$; it is not a bit-complexity bound for encoded probabilities or exact target evaluation.
If the optional pair credits are computed by enumerating ordered incident-edge pairs, their preprocessing adds $O\bigl(D\sum_u|\delta(u)|^2\bigr)$ operations.

\paragraph{Approximate guarantee.}
Applying the single-server target argument to the rounded categorical stream with no complementarity credit makes the common target $\gamma=\Gamma(D+1/m,D)$ feasible.
Indeed, $\widetilde B_u\le D+1/m$ and $\widetilde\theta_u\le D$, \wh{so Lemma~\ref{lem:gamma-monotonicity} gives $\Gamma(\widetilde B_u,\widetilde\theta_u)\ge\Gamma(D+1/m,D)$.}
We show that
\[
 \Gamma(D+1/m,D)\ge\frac{m}{m+1}R(D).
\]
Write $r=1/m$, $A=1+2D+e^{-2D}$, and $C=7D+8.056$.
For $x\ge1$, $2x+e^{-2x}\ge1+x$, so \wh{$\Gamma_1(x)$} has denominator $1+2x+e^{-2x}$.
Its denominator at $x=D+r$ satisfies \wh{$1+2(D+r)+e^{-2(D+r)}\le A+2r\le(1+r)A$,} since $A\ge2$.
The \wh{second} branch has denominator $C+6r\le(1+r)C$, since $C\ge6$.
Each branch at $(D+r,D)$ is therefore at least $1/(1+r)$ times its value at $(D,D)$, and taking their maximum proves the claim.
The acceptance-marginal identity, unchanged rewards and routing, and Lemma~\ref{lem:planning-upper} yield
\[
 \E[\ALG_{\mathrm{rounded}}]
 =\Gamma(D+1/m,D)Z_{\mathrm{LP}}
 \ge\frac{m}{m+1}R(D)\OPT.
\]

\paragraph{Recovery of the full target under positive pair credits.}
Suppose the original proposal streams satisfy $\varepsilon:=\min_u\varepsilon_u>0$, where the credits are defined in \eqref{eq:bf-pair-credit}.
Set $m=\lceil1/\varepsilon\rceil$ and use the original common target $\gamma=R(D)\le1/2$ with the rounded virtual states.
We prove its feasibility directly, without treating the original pair credit as unchanged by rounding.
Fix a server and an ordered pair $e,f$ attaining its original credit, and write
\[
 S=b_e+b_f\le1,\qquad M=\|\mathbf s_e+\mathbf s_f\|_\infty,
 \qquad \delta=\min\{S-M,S-2/3\},
 \qquad \varepsilon_u=q_{e,1}q_{f,2}\delta>0.
\]
Positivity implies $T\ge2$ and $\delta\ge\varepsilon_u$.
Consequently $M\le1-\varepsilon_u$ and
\[
 Th+2h=\frac{T+2}{2mT}\le\frac1m\le\varepsilon\le\varepsilon_u.
\]
Write $\widetilde S=\widetilde b_e+\widetilde b_f$ and $\widetilde M=\|\widetilde{\mathbf s}_e+\widetilde{\mathbf s}_f\|_\infty$.
The rounded pair remains vector feasible because $\widetilde M\le M+2h\le1$.
For its rounded compensation, defined as in \eqref{eq:bf-complementarity},
\[
 \widetilde\Psi_{\mathrm{pair}}
 =\min\{\widetilde S-\widetilde M,(\widetilde S-2/3)^+\}
 \ge\delta-2h\ge0,
\]
using $\widetilde S\ge S$ and $\widetilde M\le M+2h$.
This step requires vector feasibility, not $\widetilde S\le1$.

The first two epochs are feasible: after epoch one the nonempty mass is at most $\gamma$, leaving empty mass at least $1-\gamma\ge\gamma$.
Every feasible positive state at epoch two is fully selected, since these states precede zero and their combined mass is at most $\gamma$.
Thus the event that the first two proposals are $e,f$ and both are accepted has probability exactly $\gamma q_{e,1}q_{f,2}$.
Writing $\eta_u=\varepsilon_u-2h\ge0$, we obtain
\[
 \E\widetilde\Psi_2
 \ge\gamma q_{e,1}q_{f,2}(\delta-2h)
 \ge\gamma\eta_u.
\]
Each acceptance increases $\widetilde M$ by at most its rounded scalar size.
Both $\widetilde X-\widetilde M$ and $(\widetilde X-2/3)^+$, and hence their minimum $\widetilde\Psi$, are pathwise nondecreasing.
The same expectation bound holds at every feasible prefix ending at $t\ge2$.
\wh{For $0\le t\le T$, define $B_{u,t}=\sum_{i=1}^t\sum_{e\in\delta(u)}q_{e,i}b_e$ and $\widetilde B_{u,t}=\sum_{i=1}^t\sum_{e\in\delta(u)}q_{e,i}\widetilde b_e$ as the original and rounded expected scalar proposal loads through epoch $t$, respectively.}
For $t=0,1$, the original pair-credit argument gives $\varepsilon_u\le\psi_2\le B_u-B_{u,t}$, so
\[
 \widetilde B_u-\widetilde B_{u,t}
 \ge B_u-B_{u,t}\ge\varepsilon_u\ge\eta_u.
\]
Lemma~\ref{lem:bf-marginal} and nonnegativity of compensation therefore give, at every feasible prefix including $t=0,1$,
\[
 \E[\widetilde X_t-\widetilde\Psi_t]
 \le\gamma\bar B_u,
 \qquad \E\widetilde J_t\le\gamma\theta_u,
 \qquad \bar B_u:=\widetilde B_u-\eta_u.
\]
Also $0\le\bar B_u\le B_u+Th-\varepsilon_u+2h\le D$.
For $2\le t<T$, suppose epochs $1,\ldots,t$ are feasible.
Lemma~\ref{lem:bf-jmz} applies to the rounded stream.
If \wh{$\Gamma_2(D,D)$} attains $R(D)$, its empty-state bound gives
\[
 \Prob(\widetilde M_t=0)
 \ge1-\gamma(3\bar B_u/2+\theta_u/4+1.014)\ge\gamma.
\]
If \wh{$\Gamma_1(D)$} attains $R(D)$, monotonicity of $H_\gamma$ and $\widetilde M_t\le\widetilde X_t-\widetilde\Psi_t$ give
\[
 \Prob(\widetilde M_t=0)
 \ge H_\gamma(\gamma\bar B_u)
 \ge H_\gamma(\gamma D)\ge\gamma.
\]
In either case the empty state supplies enough mass for every legal rounded proposal, proving feasibility of the next epoch.
Induction establishes feasibility throughout the horizon.
Summing the acceptance-marginal reward identities gives $\E[\ALG_{\mathrm{rounded}}]=R(D)Z_{\mathrm{LP}}\ge R(D)\OPT$.

{
\section{Proof of Theorem~\ref{thm:limit-heavy}}
\label{app:hardness-details}}

\proof{Proof of the competitive ratio lower bound.}
{Fix a server and set {$h_t=\sum_e q_{e,t}$}, {$\Theta_t=\sum_{i\le t}h_i$}, and {$\gamma=1/(1+\theta)$}.
Epoch {$1$} is feasible because the initial state is empty and {$\gamma\le1$}.
For {$t\ge1$}, suppose epochs {$1,\ldots,t$} are feasible.
Every accepted proposal is heavy, so {$M_t>0$} if and only if {$J_t\ge1$}.
Lemma~\ref{lem:bf-marginal} and the indicator bound {$\one_{\{J_t\ge1\}}\le J_t$} give
{\[
 \Prob(M_t=0)\ge1-\E J_t
 =1-\gamma\Theta_t\ge1-\gamma\theta=\gamma.
\]}
The empty state fits every legal vector proposal, so epoch {$t+1$} is feasible.
This closes induction without assuming that at most one heavy request can be accepted.
The conditional acceptance identities give expected reward {$\gamma Z$}.
Summing over servers and using {$\theta_u\le D$} and the LP upper bound proves the claim.
This specialized target need not be at most {$1/2$}.}
\Halmos\endproof

\proof{Proof of the competitive ratio upper bound.}
{\emph{Compatible-group geometry.}}

\wh{Fix $D\ge2$ and $\zeta>0$, choose $0<p<\min\{1/(2D),\zeta/2\}$, and take an integer $n\ge1$ such that $(1-p^D)^n\le p$.}
\wh{Such a finite $n$ exists because $0<1-p^D<1$.}
There is one server \wh{$u$}, with unit capacity in each of its $D$ coordinates, and every accepted request consumes its vector irrevocably.
Choose $0<\xi<1/(2{\wh{n}}D^{\wh{n}})$.
For each ${\wh{g}}\in\{1,\ldots,{\wh{n}}\}$ and ${\wh{k}}\in\{1,\ldots,D\}$, define \wh{type $v_{g,k}$ with consumption vector $\mathbf s_{u v_{g,k}}$}
\begin{equation}
\wh{s_{u v_{g,k}}^d}=
\begin{cases}
1-\xi {\wh{g}}D^{\wh{g}},&d={\wh{k}},\\
\xi {\wh{g}}D^{{\wh{g}}-1},&d\ne {\wh{k}}.
\end{cases}
\label{eq:limit-group-vectors}
\end{equation}
The compatible-group geometry is the construction of \citet{naoriRaz2019}; we verify all of its properties needed here.
The diagonal coordinate exceeds $1/2$, and the off-diagonal coordinates are strictly positive and less than $1/(2D)$.
All $D$ different vectors in a fixed group ${\wh{g}}$ fit together, because their sum in coordinate $d$ is $1-\xi {\wh{g}}D^{{\wh{g}}-1}<1$.
Two copies of one vector do not fit, since its diagonal coordinate exceeds $1/2$.
For ${\wh{g}}<{\wh{g}}'$ and arbitrary ${\wh{k}},{\wh{k}}'$, every coordinate of $\wh{\mathbf s_{u v_{g',k'}}}$ is at least $\xi {\wh{g}}'D^{{\wh{g}}'-1}$.
In coordinate ${\wh{k}}$,
\[
\wh{s_{u v_{g,k}}^k}+\wh{s_{u v_{g',k'}}^k}
\ge1-\xi {\wh{g}}D^{\wh{g}}+\xi {\wh{g}}'D^{{\wh{g}}'-1}>1,
\]
where the strict inequality follows from ${\wh{g}}'D^{{\wh{g}}'-1}\ge({\wh{g}}+1)D^{\wh{g}}>{\wh{g}}D^{\wh{g}}$.
Thus any two vectors belonging to different groups are incompatible.

{\emph{Independent arrivals and reward bound.}}
Use ${\wh{n}}D$ early arrivals, one for each group type \wh{$v_{g,k}$, at epoch $t(g,k)=(g-1)D+k$}.
Its designated type appears with probability ${\wh{p}}$; otherwise \wh{the null type $v_{\varnothing}$, which has no compatible server,} appears.
The arrival types are independent and all active early rewards are one.
Each group is complete with probability ${\wh{p}}^D$, independently across groups.
Thus \wh{$D[1-(1-{\wh{p}}^D)^{\wh{n}}]\le\OPT_{\mathrm{early}}\le D$.}
Let $A$ be the event of any early acceptance.
At its first acceptance index, the selected group is determined by the observed history.
Only its other $D-1$ types can subsequently be accepted.
A type whose arrival time has already passed contributes zero.
Conditional on the full observed history and policy randomization through the first acceptance, each remaining type still appears at its designated arrival time with probability ${\wh{p}}$.
Independence from the past therefore gives \wh{$\E[\ALG_{\mathrm{early}}]\le[1+(D-1){\wh{p}}]\Prob(A)$.}

The final arrival, independent of the early arrivals, offers \wh{type $v_{\mathrm f}$} with vector $(1,\ldots,1)$ and reward $1/{\wh{p}}$ with probability ${\wh{p}}$; otherwise \wh{the null type arrives}.
Its expected offered reward is one.
An early acceptance blocks it because every coordinate of an early vector is positive.
Therefore
\[
 \E[\ALG]\le[1+(D-1){\wh{p}}]\Prob(A)+\Prob(A^c)\le1+(D-1){\wh{p}}.
\]
Since $1/{\wh{p}}>D$, the offline optimum takes the final request whenever it appears and otherwise takes an early optimum.
Independence gives
\[
 \OPT=1+(1-{\wh{p}})\OPT_{\mathrm{early}},\qquad
 \frac{\E[\ALG]}{\OPT}
 \le\frac{1+(D-1){\wh{p}}}{1+(1-{\wh{p}})D[1-(1-{\wh{p}}^D)^{\wh{n}}]}.
\]
\wh{Our choice of $n$ makes the ratio at most}
\[
 \wh{\frac{1+(D-1)p}{1+D(1-p)^2}
 <\frac1{D+1}+2p
 <\frac1{D+1}+\zeta.}
\]
\wh{For the first strict inequality, subtracting $1/(D+1)$ gives}
\[
 \wh{\frac{p(D^2+2D-1-Dp)}
 {(D+1)[1+D(1-p)^2]}<2p,}
\]
\wh{because $p<1/(2D)$ implies $1+D(1-p)^2>D$, and $D^2+2D-1-Dp<2D(D+1)$.}
These choices precede every policy, and $T={\wh{n}}D+1$ is a fixed finite horizon.

{For {$D=1$}, request one has size one and reward one deterministically.
Request two independently has size one and reward {$1/{\wh{p}}$} with probability {\wh{$0<p<\min\{1/2,\zeta/2\}$}}, and is otherwise incompatible.
Every online policy earns at most one, whereas {$\OPT=2-{\wh{p}}$}.
\wh{Since $1/(2-p)-1/2=p/[2(2-p)]<p<\zeta$, this gives the claimed $1/2+\zeta$ upper bound.}
Every edge in these constructions is heavy.
Multiplying all rewards by {${\wh{p}}$} places them in {$[0,1]$} and preserves ratios.}\Halmos\endproof

{\section{Proof of Theorem~\ref{thm:lp-online-limit}}
\label{app:lp-hardness}}
\proof{Proof.}
\wh{Fix $D\ge1$ and $\zeta>0$, choose $0<\omega<\min\{1/[2(D+1)],\zeta/D\}$, and set $p_\omega=1-2(D+1)\omega$.}
\wh{Consider one server with $D$ unit capacities and horizon $T=D+2$.}
\wh{At epoch $1$, a request with vector $(1/2,\omega,\ldots,\omega)$ and reward one arrives certainly.}
\wh{At epoch $2$, a request with vector $(1/2+\omega,\omega,\ldots,\omega)$ and reward one arrives with probability $p_\omega$.}
\wh{For each $d=2,\ldots,D$, epoch $d+1$ has a request with consumption one in coordinate $d$ and $\omega$ elsewhere, reward one, and arrival probability $p_\omega$.}
\wh{The final request has vector $\mathbf1$, reward $1/\omega$, and arrival probability $\omega$.}
\wh{Each listed epoch has its own type, arrivals are independent, and the remaining probability is assigned to a null type with no compatible server.}

\wh{Any two active requests are incompatible.}
\wh{The first two exceed capacity together in coordinate $1$; a request with a unit coordinate conflicts with every other request because all consumptions are positive.}
\wh{Consequently, a feasible policy can accept at most one request.}
\wh{Let $A$ be the event that an arbitrary nonanticipatory policy $\Pi$ accepts before the final epoch, and let $F$ be the final request's arrival event.}
\wh{Because arrivals are independent and the policy's internal randomization is independent of future arrivals, $A$ is independent of $F$.}
\wh{All pre-final rewards are one, so}
\[
 \wh{\E[\ALG_\Pi]
 \le \Prob(A)+\frac1\omega\Prob(A^c\cap F)
 =\Prob(A)+\Prob(A^c)=1.}
\]
\wh{Accepting the first request attains one, so the optimal online expected reward is exactly one.}

\wh{For the LP relaxation, assign each type its full expected supply.}
\wh{The first-coordinate load is}
\[
 \wh{\frac12+p_\omega\left(\frac12+D\omega\right)+\omega
 =1-2D(D+1)\omega^2\le1.}
\]
\wh{For $D\ge2$, the load in each other coordinate is}
\[
 \wh{2\omega+p_\omega[1+(D-1)\omega]
 =1-(D+1)\omega-2(D+1)(D-1)\omega^2\le1.}
\]
\wh{Every heavy-assignment count row has load $p_\omega+\omega=1-(2D+1)\omega<1$.}
\wh{The first request has size exactly $1/2$ in coordinate $1$, so it does not enter that heavy-assignment count row.}
\wh{Thus full supply is feasible; since all rewards are positive and every variable is at its supply bound, it is also optimal, with} \wh{$Z_{\mathrm{LP}}=2+Dp_\omega =D+2-2D(D+1)\omega$.}
\wh{It follows that every online policy satisfies}
\[
 \wh{\frac{\E[\ALG_\Pi]}{Z_{\mathrm{LP}}}
 \le\frac1{D+2-2D(D+1)\omega}
 \longrightarrow\frac1{D+2}
 \quad\text{as }\omega\downarrow0.}
\]
\wh{Because $D+2-2D(D+1)\omega>2$, the difference between the displayed bound and $1/(D+2)$ is}
\[
 \wh{\frac{2D(D+1)\omega}
 {(D+2)[D+2-2D(D+1)\omega]}
 <D\omega<\zeta.}
\]
\hh{Since $\zeta>0$ is arbitrary, no online policy can guarantee a ratio greater than $1/(D+2)$ relative to $Z_{\mathrm{LP}}$ uniformly over these finite-horizon instances.}
\wh{For $D=1$, the additional-coordinate requests are absent and the same argument applies.}
\wh{Multiplying all rewards by $\omega$ places them in $[0,1]$ without changing any ratio.}
\Halmos\endproof

\section{Details of the Numerical Study}
\label{app:computational-details}
This appendix supplies the data construction, implementation details, and complete numerical results for Section~\ref{sec:computations}.
\subsection{Data Preprocessing and Instance Construction}
\label{app:pom-data-construction}

\wh{For each resource, we normalize session usage by the ceiling of half its pooled median and round upward to obtain integer reservations.}
\wh{Profiles exceeding the largest specified budget in any of the four coordinates are grouped into one null-arrival type, whose probability is retained and which has no compatible server under any dimensional specification.}
\wh{With this type fixed, session-frequency-weighted $K$-medoids clustering on standardized resource features forms the remaining $K-1$ request types, giving $K$ types in total; rewards do not enter clustering.}
\wh{Each remaining type's allowance $\widetilde a_{vd}$ is its cluster's session-weighted mean reservation, rounded to the nearest integer to facilitate offline enumeration of reachable load states.}

\wh{All experiments use the same four heterogeneous service pools.}
\wh{Server budgets $c_{ud}$ use the same reservation units, giving $s_{uv}^{d}=\widetilde a_{vd}/c_{ud}$ on retained coordinates; $(u,v)$ is compatible when all these ratios are at most one.}
\wh{The catalogue is constructed using all four resources and remains fixed across horizons, dimensions, and arrival laws.}
\wh{We apply the coordinate-dominance preprocessing in} Appendix~\ref{sec:dimension-preprocessing} \wh{before evaluating the dimensional guarantee.}
\wh{After clustering, a type's value $\widetilde r_v$ is its cluster mean of raw input tokens plus $\lambda$ times raw output tokens, divided by the pooled session mean of this value.}
\wh{Its reward is $w_{uv}=\alpha_u\widetilde r_v$, where $\alpha_u$ is the server-specific reward factor.}
\wh{Budgets and reward factors are experimental inputs shared by all policies.}

Table~\ref{tab:ec-server-parameters} \hh{lists the fixed parameters used to represent differences in resource availability and service value across servers.}
\hh{For $D=2,3,4$, we retain the first $D$ resource columns; these parameters are unchanged across $K,T,S,\lambda$.}

\begin{table}[!htbp]
\centering
\caption{\hh{Server budgets $c_{ud}$ (reservation units) and reward factors $\alpha_u$ (dimensionless).}}
\label{tab:ec-server-parameters}
\small
\setlength{\tabcolsep}{6pt}
\begin{tabular}{@{}rrrrrr@{}}
\toprule
\hh{Server $u$} & \shortstack{\hh{Input tokens}\\\hh{$c_{u1}$}} & \shortstack{\hh{Output tokens}\\\hh{$c_{u2}$}} & \shortstack{\hh{Tool calls}\\\hh{$c_{u3}$}} & \shortstack{\hh{LLM calls}\\\hh{$c_{u4}$}} & \hh{$\alpha_u$} \\
\midrule
\hh{1} & \hh{4} & \hh{6} & \hh{4} & \hh{6} & \hh{0.85} \\
\hh{2} & \hh{6} & \hh{10} & \hh{12} & \hh{6} & \hh{0.95} \\
\hh{3} & \hh{10} & \hh{4} & \hh{6} & \hh{12} & \hh{1.05} \\
\hh{4} & \hh{12} & \hh{12} & \hh{10} & \hh{8} & \hh{1.15} \\
\bottomrule
\end{tabular}
\end{table}

\wh{The arrival probabilities $p_{vt}$ are the type frequencies in equal-probability-mass segments of the saved session sequence, splitting boundary-session mass when necessary.}
\wh{We form ten segments, pool adjacent pairs for $S=5$, and pool all segments for $S=1$.}
\wh{Each segment spans $T/S$ epochs with independent type draws; the three laws therefore preserve expected type counts at fixed $(K,D,T)$.}

\subsection{VABF-T and VABF-TR: Tuning and Recovery}
\label{app:experimental-policies}
\label{app:bfr-implementation}

\paragraph{VABF-T.}
\wh{VABF-T keeps the LP solution and acceptance rule from} Section~\ref{sec:routing-admission} \wh{and changes routing.}
\wh{Conditional on the arrival of a type with positive LP allocation, the candidate proposes it to a server with probability one, preserving the LP's relative routing shares:}
\[
 r^\star_{uv}=
 \begin{cases}
 z_{uv}/\sum_jz_{jv},&\sum_jz_{jv}>0,\\
 0,&\text{otherwise}.
 \end{cases}
 \qquad q^\star_{uv,t}=p_{vt}r^\star_{uv}.
\]
\wh{For each $(K,D,T)$, we hold the LP solution and candidate routing probabilities fixed across $S=1,5,10$, recomputing state distributions and targets for each arrival law.}

For a candidate stream, set $B^\star_u=\sum_{t,v}b_{uv}q^\star_{uv,t}$ and $\theta^\star_u=\sum_{t,v:b_{uv}>1/2}q^\star_{uv,t}$.
\wh{Using $\Gamma$ in \eqref{eq:adaptive-target}, calibration takes eight bisection steps for each server between $\Gamma(B^\star_u,\theta^\star_u)-10^{-12}$ and one.}
\wh{Each trial recomputes the state distributions and checks whether feasible mass reaches the trial target for every positive-probability proposal at every epoch; infeasible trials are rejected.}
\wh{The selected constant target $\gamma_u^\star$ is checked again over the full horizon.}
\wh{By Lemma~\ref{lem:bf-marginal}, the candidate's expected reward in exact arithmetic is $\sum_{u,t,v}\gamma_u^\star q^\star_{uv,t}w_{uv}$.}
\wh{Before arrivals, VABF-T compares the candidate's total expected reward with that of VABF.}
\wh{It adopts the candidate only when numerical error bounds verify no decrease in expected reward and preservation of the theoretical reward floor; otherwise it retains VABF.}

\paragraph{VABF-TR.}
\wh{VABF-TR uses the selected VABF-T policy as a shadow policy, sharing its arrivals, routing draws, and acceptance coins.}
\wh{For a fixed server, suppress $u$ and use $\mathbf s_e$, $q_{e,t}$, and $a_t(\boldsymbol\ell,e)$ as in} Section~\ref{sec:routing-admission}; \wh{$\boldsymbol\ell$ is the shadow state and $\mathbf y$ the actual load, including recovery assignments.}
\wh{Offline backward recursion gives a safe limit $\mathbf H_t(\boldsymbol\ell)$ on actual load that reserves capacity for every possible sequence of later shadow assignments:}
\begin{align}
 \wh{\mathbf H_{T+1}(\boldsymbol\ell)}&=\wh{\mathbf1}, \notag\\
 \wh{\mathbf H_t(\boldsymbol\ell)}&=
 \wh{\min\left\{\mathbf H_{t+1}(\boldsymbol\ell),
 \min_{e:\,q_{e,t}a_t(\boldsymbol\ell,e)>0}
 \big[\mathbf H_{t+1}(\boldsymbol\ell+\mathbf s_e)-\mathbf s_e\big]\right\}},
 \label{eq:bfr-envelope}
\end{align}
\wh{All minima are coordinatewise, and an empty inner minimum imposes no constraint.}
\wh{The two terms cover no assignment and acceptance of $e$ by the shadow policy, respectively.}
\wh{The recursion uses the known proposal distributions and acceptance rule, without observing future arrivals.}

\wh{The following conditions are applied separately at each server using its own actual load, shadow state, consumption vector, and capacity limit.}
\wh{Before epoch $t$, maintain $\mathbf y\le\mathbf H_t(\boldsymbol\ell)$.}
\wh{A shadow acceptance is implemented at the same server and adds $\mathbf s_e$ to both loads.}
\wh{If the shadow policy rejects the request, examine every compatible server and select the highest-reward assignment satisfying $\mathbf y+\mathbf s_e\le\mathbf H_{t+1}(\boldsymbol\ell)$, breaking ties by server index; reject if none qualifies.}
\wh{Recovery adds $\mathbf s_e$ only to $\mathbf y$, leaving the shadow policy and its state unchanged.}

\wh{Backward induction in \eqref{eq:bfr-envelope} gives $\boldsymbol\ell\le\mathbf H_t(\boldsymbol\ell)\le\mathbf1$, so zero initial loads satisfy the condition; the recursion and these updates preserve it at every epoch.}
\wh{All shadow assignments therefore remain feasible, and nonnegative rewards together with VABF-T's selection rule give}
\[
 \wh{\E[\text{VABF-TR}(I)]\ge\E[\text{VABF-T}(I)]\ge\E[\text{VABF}(I)]}.
\]
\wh{Both policies thus inherit VABF's performance guarantee.}

\subsection{Comparators, Simulation, and Verification}
\label{app:experimental-verification}
RJMZ solves its own maximum-coordinate scalar-surrogate LP and applies scalar Best-Fit with target $(3+e^{-2})^{-1}$, as in Corollary~\ref{cor:scalar-jmz}.
\wh{It retains the scalar LP's routing probabilities and assigns a common acceptance probability to feasible boundary states.}
NR implements Algorithm~1 of \citet{naoriRaz2019}: an initial sampling phase is followed by observed-prefix heavy matching and \wh{randomized rounding of an LP for smaller requests}.
We use $q_2=2D/(2D+1)$, $q_1=q_2/e^{1/4}$, and floor-rounded phase boundaries.
Incompatible edges are masked; unserviceable requests retain their positions in the horizon.
Its random-order premise applies to i.i.d.\ arrivals conditional on the realized multiset; no such guarantee is invoked for $S=5,10$.

There are $2\times3\times6\times3=108$ cells.
\wh{Each main-study instance uses 64 paths generated from shared arrival, routing, acceptance, and item-priority uniform banks.}
\wh{The reward-weight comparison uses the first five paths from the same banks.}
Across $S$, a shared arrival uniform is transformed using the appropriate segment distribution.
\wh{Average rewards use the common vector-LP denominator defined in} Section~\ref{sec:exp-design}, \wh{including for RJMZ, whose scalar LP determines its decisions.}

\wh{Validation checks the nested arrival laws, acceptance marginals and ordering, assignments, rewards, and capacity use.}
\wh{The study uses small enumerated state graphs, with at most ten states in the main grid.}
The implementation uses Python 3.9.6, NumPy 1.26.3, and SciPy 1.11.4, with one BLAS thread.
\wh{The reproducibility files retain the fixed catalogues, shared random banks, and validation outputs.}

\FloatBarrier
\subsection{Guarantee-Preserving Heuristics: Complete Results}
\label{app:complete-experiment-results}
Tables~\ref{tab:pom-full-heuristics}--\ref{tab:pom-full-heuristics-s10} \wh{report the average rewards of VABF, VABF-T, and VABF-TR for all 108 main-study instances at $\lambda=3$.}
\wh{Each table places $K=10$ and $K=20$ side by side for all three resource configurations and all six demand levels.}
\wh{Rewards are expressed as percentages of each instance's numerically computed $Z_{\mathrm{LP}}$.}
\wh{The $K=20,S=5,T=100,200,300$ entries reproduce the average rewards in} Table~\ref{tab:pom-enhancements}.

\begin{table}[!htbp]
\centering
\caption{\wh{Average rewards under pooled i.i.d. arrivals, $S=1$ and $\lambda=3$ (reward / $Z_{\mathrm{LP}}$, \%).}}
\label{tab:pom-full-heuristics}
\label{tab:pom-full-heuristics-s1}
\label{tab:pom-full-k20-s1}
\label{tab:pom-full-k10-s1}
\small
\setlength{\tabcolsep}{5pt}
\begin{tabular}{@{}rrrrr@{\hspace{16pt}}rrr@{}}
\toprule
 & & \multicolumn{3}{c}{\wh{$K=10$}} & \multicolumn{3}{c}{\wh{$K=20$}} \\
\cmidrule(lr){3-5}\cmidrule(lr){6-8}
$D$ & $T$ & \wh{VABF} & \wh{VABF-T} & \wh{VABF-TR} & \wh{VABF} & \wh{VABF-T} & \wh{VABF-TR} \\
\midrule
2 & 50 & 32.79 & 58.93 & 65.58 & 30.91 & 56.58 & 63.51 \\
 & 100 & 35.23 & 75.34 & 80.09 & 34.20 & 60.13 & 67.21 \\
 & 150 & 35.79 & 77.82 & 82.76 & 38.12 & 67.79 & 75.20 \\
 & 200 & 36.08 & 77.90 & 83.38 & 36.32 & 75.25 & 80.17 \\
 & 250 & 39.93 & 81.57 & 86.56 & 36.13 & 74.43 & 80.11 \\
 & 300 & 37.43 & 77.36 & 82.67 & 33.91 & 77.35 & 82.68 \\
\addlinespace[2pt]
3 & 50 & 32.25 & 58.42 & 65.93 & 30.83 & 55.66 & 62.07 \\
 & 100 & 33.89 & 72.54 & 78.86 & 33.98 & 59.02 & 65.80 \\
 & 150 & 33.74 & 71.37 & 78.08 & 37.27 & 65.78 & 73.05 \\
 & 200 & 32.94 & 75.49 & 81.96 & 36.45 & 73.90 & 78.70 \\
 & 250 & 37.18 & 75.83 & 82.27 & 35.63 & 73.55 & 79.34 \\
 & 300 & 36.52 & 71.91 & 78.96 & 33.63 & 75.79 & 80.90 \\
\addlinespace[2pt]
4 & 50 & 32.21 & 58.05 & 67.19 & 29.75 & 54.74 & 59.86 \\
 & 100 & 32.91 & 68.97 & 79.06 & 30.67 & 56.54 & 64.08 \\
 & 150 & 33.71 & 70.64 & 80.64 & 32.99 & 63.34 & 72.06 \\
 & 200 & 30.52 & 70.33 & 81.19 & 32.46 & 66.38 & 74.33 \\
 & 250 & 31.05 & 73.52 & 84.52 & 30.28 & 66.29 & 74.33 \\
 & 300 & 32.52 & 69.76 & 80.08 & 33.52 & 67.20 & 74.22 \\
\bottomrule
\end{tabular}
\end{table}

\begin{table}[!htbp]
\centering
\caption{\wh{Average rewards under nonstationary arrivals, $S=5$ and $\lambda=3$ (reward / $Z_{\mathrm{LP}}$, \%).}}
\label{tab:pom-full-heuristics-s5}
\label{tab:pom-full-k20-s5}
\label{tab:pom-full-k10-s5}
\small
\setlength{\tabcolsep}{5pt}
\begin{tabular}{@{}rrrrr@{\hspace{16pt}}rrr@{}}
\toprule
 & & \multicolumn{3}{c}{\wh{$K=10$}} & \multicolumn{3}{c}{\wh{$K=20$}} \\
\cmidrule(lr){3-5}\cmidrule(lr){6-8}
$D$ & $T$ & \wh{VABF} & \wh{VABF-T} & \wh{VABF-TR} & \wh{VABF} & \wh{VABF-T} & \wh{VABF-TR} \\
\midrule
2 & 50 & 30.27 & 60.72 & 66.61 & 32.02 & 50.06 & 57.22 \\
 & 100 & 29.93 & 72.89 & 78.25 & 30.49 & 64.53 & 72.25 \\
 & 150 & 31.00 & 75.47 & 80.63 & 32.17 & 70.02 & 76.50 \\
 & 200 & 31.51 & 80.09 & 85.45 & 32.46 & 71.30 & 77.89 \\
 & 250 & 31.29 & 80.79 & 85.53 & 37.98 & 77.79 & 82.76 \\
 & 300 & 30.60 & 77.05 & 82.23 & 35.36 & 74.08 & 79.88 \\
\addlinespace[2pt]
3 & 50 & 31.54 & 60.31 & 67.16 & 31.51 & 50.40 & 57.45 \\
 & 100 & 29.88 & 68.98 & 75.35 & 30.58 & 61.63 & 69.51 \\
 & 150 & 29.48 & 72.08 & 78.57 & 31.80 & 68.98 & 75.24 \\
 & 200 & 30.80 & 78.18 & 84.05 & 32.63 & 70.83 & 77.37 \\
 & 250 & 29.26 & 75.82 & 81.53 & 37.58 & 76.58 & 81.54 \\
 & 300 & 27.94 & 71.75 & 78.02 & 34.42 & 72.89 & 78.71 \\
\addlinespace[2pt]
4 & 50 & 32.29 & 58.73 & 67.56 & 32.89 & 51.47 & 57.02 \\
 & 100 & 30.62 & 66.65 & 75.75 & 28.79 & 58.28 & 67.31 \\
 & 150 & 32.00 & 70.38 & 79.93 & 30.02 & 62.93 & 70.59 \\
 & 200 & 32.30 & 73.78 & 83.45 & 32.30 & 62.89 & 71.12 \\
 & 250 & 26.29 & 71.06 & 81.52 & 32.25 & 69.42 & 76.47 \\
 & 300 & 27.65 & 67.67 & 77.99 & 32.80 & 68.67 & 77.35 \\
\bottomrule
\end{tabular}
\end{table}

\begin{table}[!htbp]
\centering
\caption{\wh{Average rewards under nonstationary arrivals, $S=10$ and $\lambda=3$ (reward / $Z_{\mathrm{LP}}$, \%).}}
\label{tab:pom-full-heuristics-s10}
\label{tab:pom-full-k20-s10}
\label{tab:pom-full-k10-s10}
\small
\setlength{\tabcolsep}{5pt}
\begin{tabular}{@{}rrrrr@{\hspace{16pt}}rrr@{}}
\toprule
 & & \multicolumn{3}{c}{\wh{$K=10$}} & \multicolumn{3}{c}{\wh{$K=20$}} \\
\cmidrule(lr){3-5}\cmidrule(lr){6-8}
$D$ & $T$ & \wh{VABF} & \wh{VABF-T} & \wh{VABF-TR} & \wh{VABF} & \wh{VABF-T} & \wh{VABF-TR} \\
\midrule
2 & 50 & 33.76 & 59.20 & 65.59 & 30.37 & 50.28 & 57.93 \\
 & 100 & 31.48 & 72.07 & 77.62 & 29.80 & 61.21 & 67.72 \\
 & 150 & 31.73 & 77.42 & 82.58 & 38.25 & 74.38 & 80.10 \\
 & 200 & 34.56 & 79.00 & 83.86 & 36.76 & 69.64 & 75.72 \\
 & 250 & 36.51 & 78.81 & 83.16 & 40.94 & 77.53 & 82.47 \\
 & 300 & 30.73 & 76.78 & 81.95 & 33.75 & 78.23 & 82.49 \\
\addlinespace[2pt]
3 & 50 & 33.48 & 59.55 & 66.64 & 29.40 & 50.20 & 57.71 \\
 & 100 & 31.51 & 68.21 & 74.62 & 29.17 & 60.08 & 67.01 \\
 & 150 & 29.78 & 73.31 & 79.84 & 37.35 & 72.93 & 78.44 \\
 & 200 & 32.09 & 74.02 & 80.57 & 36.54 & 68.38 & 74.34 \\
 & 250 & 33.51 & 74.42 & 80.47 & 41.15 & 75.18 & 79.96 \\
 & 300 & 29.06 & 70.97 & 77.57 & 33.33 & 76.26 & 80.57 \\
\addlinespace[2pt]
4 & 50 & 31.56 & 57.78 & 68.06 & 29.99 & 53.65 & 60.13 \\
 & 100 & 31.80 & 66.06 & 75.84 & 29.13 & 58.73 & 67.74 \\
 & 150 & 32.44 & 70.19 & 80.63 & 35.05 & 64.31 & 72.49 \\
 & 200 & 33.72 & 72.93 & 82.35 & 33.01 & 63.94 & 71.57 \\
 & 250 & 31.39 & 70.69 & 81.87 & 39.43 & 71.26 & 78.76 \\
 & 300 & 27.97 & 68.17 & 79.31 & 34.31 & 69.33 & 76.46 \\
\bottomrule
\end{tabular}
\end{table}

\FloatBarrier
\wh{For the reward-weight comparison, we vary $\lambda=1,3,5,10$ and recompute type rewards as in} Appendix~\ref{app:pom-data-construction}, \wh{using the pooled session mean at each weight for normalization.}
Cluster memberships, resource representatives, budgets, and arrival probabilities are held fixed.
\wh{Each reward instance re-solves the LPs and recomputes the policies; VABF and VABF-T share the same vector LP solution.}
\wh{At each weight, we cross $K=10,20$, $D=2,3,4$, $T=100,200,300$, and $S=1,5,10$.}
Table~\ref{tab:pom-heuristic-reward} \wh{reports the average normalized rewards over these 54 settings.}

\wh{The average rewards of VABF-T and VABF-TR vary modestly across the four reward weights, while individual instances can be more sensitive.}
For \wh{$K=20,D=3,T=200,S=1$}, the \wh{average rewards for VABF-T} at $\lambda=1,3,5,10$ are \wh{58.92\%}, \wh{58.99\%}, \wh{49.76\%}, and \wh{62.76\%}, respectively.
\wh{VABF} and \wh{VABF-T} respond through the vector LP, routing, and acceptance targets, and \wh{VABF-TR} also recomputes the \wh{capacity limits for its shadow policy}.
\wh{These comparisons describe performance after the LP and policies are re-optimized for the changed reward weights.}

\begin{table}[!htbp]
\centering
\caption{\wh{Average rewards across reward weights (reward / $Z_{\mathrm{LP}}$, \%).}}
\label{tab:pom-heuristic-reward}
\small
\setlength{\tabcolsep}{6pt}
\begin{tabular}{@{}rrrr@{}}
\toprule
$\lambda$ & \wh{VABF} & \wh{VABF-T} & \wh{VABF-TR} \\
\midrule
1 & 37.15 & 71.10 & 77.96 \\
3 & 37.07 & 70.95 & 78.26 \\
5 & 34.90 & 69.20 & 77.06 \\
10 & 33.32 & 70.90 & 77.65 \\
\bottomrule
\end{tabular}
\par\smallskip{\footnotesize\wh{Entries equally average the normalized mean rewards over 54 $(K,D,T,S)$ settings at each weight, using each instance's own $Z_{\mathrm{LP}}$.}}
\end{table}

\FloatBarrier
\subsection{Benchmark Comparisons}
\label{app:benchmark-results}
\label{app:temporal-comparisons}
\label{app:reward-sensitivity}
Figures~\ref{fig:ec-benchmarks-T-k10} and~\ref{fig:ec-benchmarks-T-k20} \wh{compare VABF-TR with RJMZ, NR, and Greedy for all 108 main-study instances.}
\wh{Each panel holds $K,D,S$ fixed and varies $T$, extending the demand-level comparison in} Figure~\ref{fig:bfr-benchmarks-T}.
\wh{Across the complete grid, VABF-TR achieves higher average reward than Greedy in most settings and than RJMZ and NR throughout.}

\begin{figure}[!htbp]
\centering
\includegraphics[width=\textwidth]{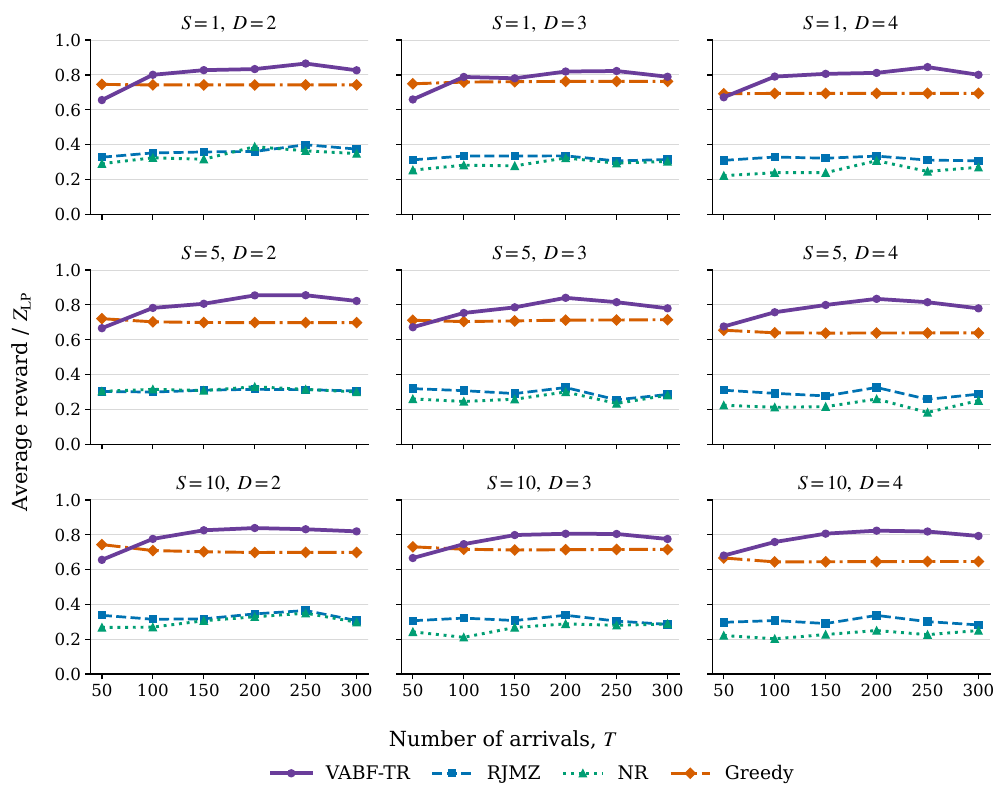}
\caption{\wh{Benchmark performance across demand levels for $K=10$ and $\lambda=3$. Rows correspond to $S=1,5,10$ and columns to $D=2,3,4$. Points report mean reward divided by each instance's $Z_{\mathrm{LP}}$; pool budgets remain fixed as $T$ increases.}}
\label{fig:ec-benchmarks-T-k10}
\label{tab:pom-full-benchmarks}
\end{figure}

\begin{figure}[!htbp]
\centering
\includegraphics[width=\textwidth]{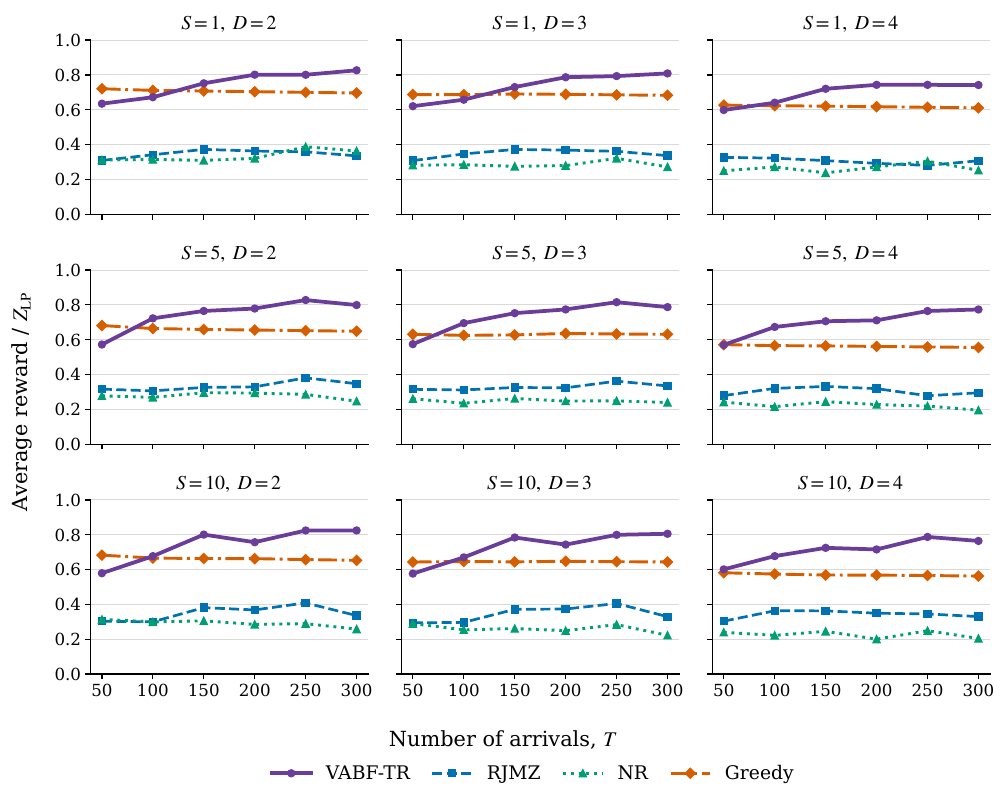}
\caption{\wh{Benchmark performance across demand levels for $K=20$ and $\lambda=3$, with the same layout and normalization as Figure~\ref{fig:ec-benchmarks-T-k10}. The $S=5$ row reproduces the values in Figure~\ref{fig:bfr-benchmarks-T}.}}
\label{fig:ec-benchmarks-T-k20}
\end{figure}

Figure~\ref{fig:ec-benchmarks-S} \wh{compares arrival laws for both catalogues, using the same equal average over six demand levels as} Figure~\ref{fig:bfr-benchmarks-S}.
\wh{The $K=10$ results show the same pattern: VABF-TR's average rewards remain similar across the three laws, while Greedy's are lower under both nonstationary laws.}

\begin{figure}[!htbp]
\centering
\includegraphics[width=.94\textwidth]{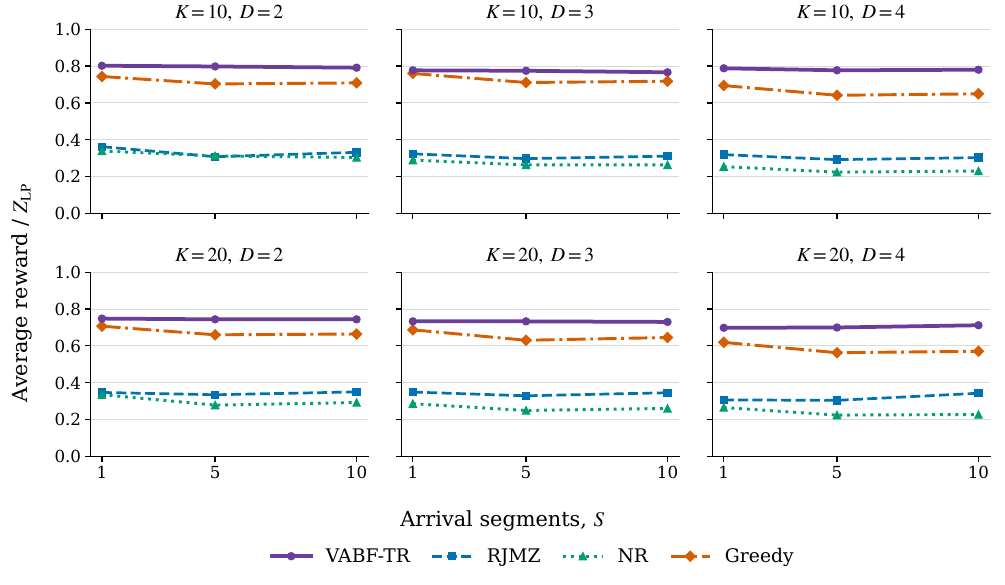}
\caption{\wh{Benchmark performance under pooled i.i.d.\ ($S=1$) and nonstationary ($S=5,10$) arrivals, $\lambda=3$. Rows correspond to $K=10,20$ and columns to $D=2,3,4$. Each point equally averages normalized mean rewards over $T=50,100,150,200,250,300$. At each demand level, the laws share expected type counts, budgets, rewards, and the LP solution.}}
\label{fig:ec-benchmarks-S}
\label{tab:pom-benchmark-sensitivity}
\label{tab:pom-temporal-avg}
\end{figure}

\Needspace{9\baselineskip}
\begin{samepage}
Figure~\ref{fig:bfr-benchmarks-lambda} \wh{compares reward weights separately for each $(K,D)$, averaging over the three demand levels and three arrival laws.}
\wh{Equally averaging its six panels recovers the overall comparison across 54 settings at each weight: VABF-TR averages 77.06\%--78.26\% of $Z_{\mathrm{LP}}$, compared with 65.66\%--66.60\% for Greedy.}
\wh{The ranking of these overall averages is unchanged across the four weights.}
\wh{In this construction, $w_{uv}(\lambda)=\alpha_u\widetilde r_v(\lambda)$ with positive type rewards, so Greedy's feasible-pool ranking and hence its assignments are unchanged along a fixed arrival path.}
\wh{Its normalized reward can nevertheless change because both the reward weights and the LP denominator change.}
\end{samepage}

\begin{figure}[!htbp]
\centering
\includegraphics[width=.94\textwidth]{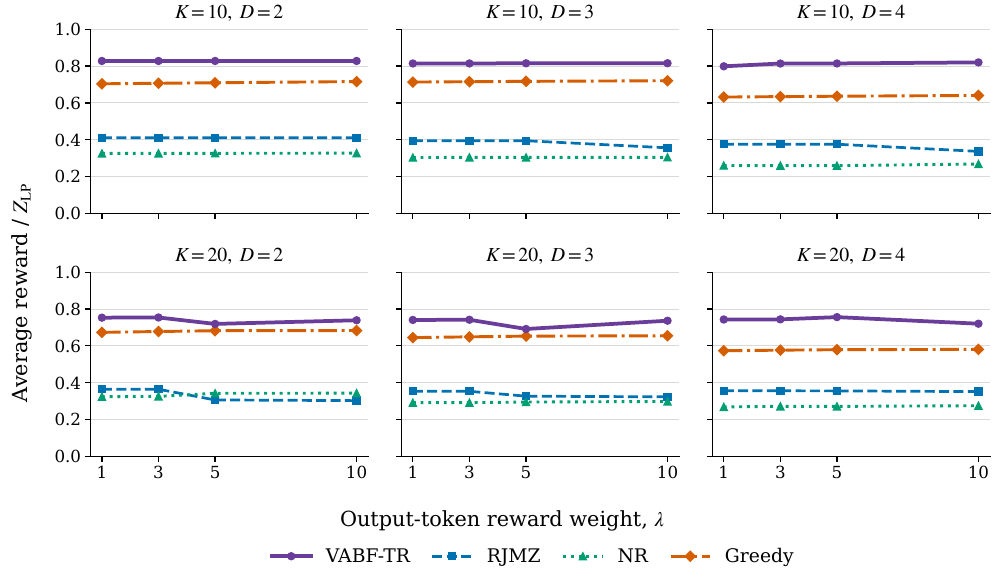}
\caption{\wh{Benchmark performance across output-token reward weights. Rows correspond to $K=10,20$ and columns to $D=2,3,4$. Each point equally averages normalized mean rewards over the nine $(T,S)$ settings with $T=100,200,300$ and $S=1,5,10$, using each instance's own $Z_{\mathrm{LP}}$. The resource catalogue is fixed, and the LPs and policies are recomputed at each weight.}}
\label{fig:bfr-benchmarks-lambda}
\label{tab:pom-reward-pilot}
\end{figure}

\FloatBarrier


\end{document}